\documentclass[journal]{IEEEtran}

\usepackage{amsmath,amsfonts,amsthm,amssymb,optidef,stmaryrd,dsfont}
\usepackage{mdframed}
\usepackage{algorithm}
\usepackage{algpseudocode}
\usepackage{array}
\usepackage[caption=false,font=normalsize,labelfont=sf,textfont=sf]{subfig}
\usepackage{textcomp}
\usepackage{stfloats}
\usepackage{url}
\usepackage{verbatim}
\usepackage{graphicx}
\usepackage{cite}
\usepackage{color}
\usepackage{multirow,booktabs,makecell}
\usepackage[hidelinks]{hyperref}

\algdef{SE}[SUBALG]{Indent}{EndIndent}{}{\algorithmicend\ }%
\algtext*{Indent}
\algtext*{EndIndent}

\mdfsetup{%
innertopmargin=0pt,
innerleftmargin=3pt,
innerbottommargin=3pt,
innerrightmargin=3pt
}

\usepackage[all=normal,paragraphs=tight,floats=normal,mathspacing=normal,wordspacing=tight,charwidths=tight,mathdisplays=normal,leading=normal]{savetrees}

\definecolor{mygreen}{RGB}{11, 218, 81}
\definecolor{myblue}{RGB}{31, 81, 255}
\definecolor{mypurple}{RGB}{191, 64, 191}
\definecolor{myorange}{RGB}{255, 172, 28}
\definecolor{myred}{RGB}{220, 38, 38} 

\newcommand{\Redit}[1]{\textcolor{black}{#1}}
\newcommand{\Rours}[1]{\textcolor{black}{#1}}

\newcommand{\Review}[1]{\textcolor{black}{#1}}

\newcommand{\abs}[1]{\ensuremath{\left| #1 \right|}}
\newcommand{\esp}[1]{\ensuremath{\mathbb{E}\left[ #1 \right]}}
\newcommand{\espo}[2]{\ensuremath{\mathbb{E}_{ #2 }\left[ #1 \right]}}
\newcommand{\norm}[2]{\ensuremath{\left\lVert #1 \right\rVert}_#2}

\newcommand{\tr}[1]{\ensuremath{\mathsf{Tr}\left( #1 \right)}}

\newcommand*{\herm}{{\mathsf{H}}}
\newcommand{\simil}[1]{\ensuremath{\textnormal{sim}\left( #1 \right)}}

\newcommand{\ftheta}{\ensuremath{f_{\boldsymbol{\theta}}}}
\newcommand{\bdelta}{\ensuremath{\boldsymbol{\delta}}}
\newcommand{\Gn}{\ensuremath{\mathbb{G}}}
\newcommand{\Gt}{\ensuremath{\mathbb{G}_{\mathsf{G}}}}
\newcommand{\Gl}{\ensuremath{\mathbb{G}_{\mathsf{L}}}}

\newcommand{\Gc}{\ensuremath{\mathbb{G}_{\mathsf{C}}}}
\newcommand{\Sone}{\ensuremath{\mathsf{S}_{1}}}
\newcommand{\Stwo}{\ensuremath{\mathsf{S}_{2}}}
\newcommand{\Sthree}{\ensuremath{\mathsf{S}_{3}}}

\usepackage{xstring} 

\newcommand{\ej}[2][+]{%
    \ensuremath{\mathrm{e}^{%
        \IfStrEq{#1}{-}{-\mathrm{j}}{\mathrm{j}} #2 }%
    }}%

\DeclareMathOperator*{\argmax}{arg\,max}
\DeclareMathOperator*{\argmin}{arg\,min}
\newcommand{\argmink}{\argmin\nolimits_{\raisebox{0ex}{$\scriptstyle k$}}}

\newtheorem{remark}{Remark}[section]

\newtheorem{definition}{Definition}

\newtheorem{theorem}{Theorem}
\newtheorem{corollary}{Corollary}

\algnewcommand\algorithmicforeach{\textbf{for each}}
\algdef{S}[FOR]{ForEach}[1]{\algorithmicforeach\ #1\ \algorithmicdo}

\begin{document}

\title{\Review{Model-based Implicit Neural Representation for\\ sub-wavelength Radio Localization}}

\author{
	Baptiste Chatelier, Vincent Corlay, Musa Furkan Keskin, Matthieu Crussière,\\ Henk Wymeersch, Luc Le Magoarou
    \thanks{Baptiste Chatelier is with Mitsubishi Electric R\&D Centre Europe, Univ. Rennes, INSA Rennes, CNRS, IETR-UMR 6164 and b\raisebox{0.2mm}{\scalebox{0.7}{\textbf{$<>$}}}com, Rennes, France (email: baptiste.chatelier@insa-rennes.fr).}%
    \thanks{Vincent Corlay is with Mitsubishi Electric R\&D Centre Europe and b\raisebox{0.2mm}{\scalebox{0.7}{\textbf{$<>$}}}com, Rennes, France (email: v.corlay@fr.merce.mee.com).}
    \thanks{Musa Furkan Keskin and Henk Wymeersch are with the Department of Electrical Engineering, Chalmers University of Technology, Gothenburg, Sweden (email: \{furkan ; henkw\}@chalmers.se).}%
    \thanks{Matthieu Crussière and Luc Le Magoarou are with Univ. Rennes, INSA Rennes, CNRS, IETR-UMR 6164 and b\raisebox{0.2mm}{\scalebox{0.7}{\textbf{$<>$}}}com, Rennes, France (email: \{matthieu.crussiere ; luc.le-magoarou\}@insa-rennes.fr).}%
    \thanks{This work has been supported by the SNS JU project 6G-DISAC under the EU's Horizon Europe Research and Innovation Program under Grant Agreement No 101139130, and the Swedish Research Council (VR) through the project 6G-PERCEF under Grant 2024-04390.}
	}



\maketitle

\begin{abstract}
    The increasing deployment of large antenna arrays at base stations has significantly improved the spatial resolution and localization accuracy of radio-localization methods. However, traditional signal processing techniques struggle in complex radio environments, particularly in scenarios dominated by non line of sight (NLoS) propagation paths, resulting in degraded localization accuracy. Recent developments in machine learning have facilitated the development of machine learning-assisted localization techniques, enhancing localization accuracy in complex radio environments. However, these methods often involve substantial computational complexity during both the training and inference phases. This work extends the well-established fingerprinting-based localization framework by simultaneously reducing its memory requirements and improving its accuracy. Specifically, a model-based neural network is used to learn the location-to-channel mapping, and then serves as a generative neural channel model. This generative model augments the fingerprinting comparison dictionary while reducing the memory requirements. The proposed method outperforms fingerprinting baselines by achieving sub-wavelength localization accuracy, even in \Review{complex static NLoS environments}. Remarkably, it offers an improvement by several orders of magnitude in localization accuracy, while simultaneously reducing memory requirements by an order of magnitude compared to classical fingerprinting methods.
\end{abstract}


\section{Introduction}\label{sec:introduction}
    \IEEEPARstart{R}{adio} localization refers to the process of determining the location of an object, such as a transmitter, a receiver or a passive target, based on measurements of radio frequency electromagnetic waves. One of the earliest application of radio localization is radio detection and ranging (radar), which has been extensively used in both military, e.g. detecting adversarial aircrafts, and civilian contexts, e.g. air traffic control in civil aviation. Traditional radio localization techniques rely on the use of signal processing methods such as angle of arrival, time of arrival, time difference of arrival or round trip time: see~\cite{4343996} for an extensive review. These methods typically require multiple sensing nodes and assume the existence of a direct line of sight (LoS) propagation path with negligible non line of sight (NLoS) paths treated as nuisance. However, in complex radio environments, NLoS paths may be significant, resulting in degraded localization accuracy. 
    
    The evolution towards the sixth-generation of cellular networks has introduced the concept of integrated sensing and communication, wherein sensing tasks, such as localization, and communication functions are jointly developed to improve system efficiency and performance. \Review{This places the localization problematic at the core of future cellular communication systems, as these methods will be extensively utilized to enable a wide range of localization-based applications~\cite{6924849,9215972}.} From this perspective, it is crucial to develop localization methods that achieve high precision in complex radio propagation environments, to maximize the efficiency of future applications. One example of a radio-localization-aided application is the deployment of automated guided vehicles in smart factories, where precise positioning is essential for efficient and autonomous operations. \Review{Smart factories are typically complex indoor environments with challenging propagation conditions, i.e. quasi-static channel but with significant NLoS paths, where classical localization methods often suffer from degraded performance.}

    Over the past decade, rapid developments in the machine learning (ML) field have facilitated the development of learning-based localization methods, particularly those leveraging fingerprinting. These methods exploit patterns in a radio-environment related feature, e.g. received signal strength or channel coefficients, to perform localization by associating a given measured feature with a position, through the use of a ML method. Such approach has been studied by the $3$GPP~\cite{3gpp_ai_ml_2021,3gpp_ai_ml_positioning_2022}. In this context, it was reported that, in complex indoor radio environments, the $90\%$ quantile localization error metric could be reduced from approximately $15$ m with classical methods to within the range of $1-5$ m when leveraging ML-aided techniques~\cite{10188249}. However, these approaches still rely on channel estimates from multiple base stations, and the proposed neural architectures often introduce significant computational complexity during both training and inference phases, in addition to requiring large dataset sizes. This paper investigates to what extent the model-based machine learning paradigm~\cite{10056957,SIG113}, identified as a key strategy for improving the efficiency and interpretability of ML techniques, can be leveraged for effective radio localization. Given the well-known challenges of localization in complex radio environments, several key questions arise: \textit{Does a model-based machine learning localization method achieve superior performance in challenging radio environments compared to classical methods? Does it entail a lower computational complexity than other machine-learning methods?}

    \noindent \textbf{Contributions.} Our previous studies~\cite{chatelier_loc2chan_icassp24,chatelier_loc2chan24} introduced a novel neural architecture that combines the concept of implicit neural representation (see~\cite{INR_review} for an extensive literature review), with the model-based machine learning paradigm, in order to learn the location-to-channel mapping at the wavelength level. The proposed neural network in~\cite{chatelier_loc2chan24} is able to infer the channel coefficients between a base station (BS) equipped with multiple antennas operating at several frequencies, and a mono-antenna user equipment (UE) at a given location, \Redit{achieving high accuracy~\footnote{\Redit{For instance, the proposed trained network achieves a $-29$ dB normalized mean squared error when inferring channel coefficients in a ray-tracing simulated outdoor scenario.}}}. Once trained, this neural network can be used as a generative neural channel model on a given scene. This work investigates the use of this neural model to perform localization with sub-wavelength precision. The main contributions of this paper are as follows.
    \begin{itemize}
        \item A localization method taking advantage of the learned location-to-channel mapping to perform data augmentation in a fingerprinting framework is introduced in Definition~\ref{def:estimator}, and is illustrated by Fig.~\ref{fig:general_synoptic}.
        \item The performance and computational complexity of the proposed method are optimized in Section~\ref{subsec:perf_optim} using theoretical insights introduced in Corollary~\ref{corol:minima_distance}.
        \item Extensive experiments are conducted on realistic channels, showing the sub-wavelength precision of the proposed method and its computational efficiency. Compared to classical fingerprinting methods, the proposed approach improves the median localization accuracy by factors ranging from $10^2$ to $10^3$, while concurrently reducing memory requirements by a factor $10$. This is illustrated in Fig.~\ref{fig:memory_performance} where the proposed method is compared against the classical $k$-Nearest Neighbors ($k$-NN) fingerprinting method.
    \end{itemize}

    \begin{figure*}[t]
        \centering
        \begin{minipage}[b]{0.49\textwidth}
            \centering
            \includegraphics[scale=.78]{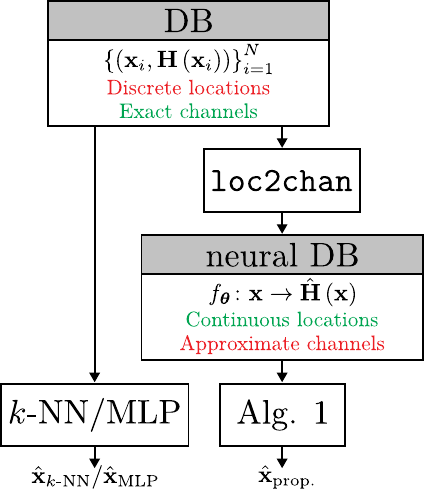}
            \caption{\Review{Overview of the proposed method: the left-hand side illustrates the classical $k$-NN/MLP approaches using a fixed database (DB), whereas the right-hand side depicts the proposed scheme, in which a neural network learning the location-to-channel mapping serves as a neural database.}}
            \label{fig:general_synoptic}
        \end{minipage}
        \hfill
        \begin{minipage}[b]{0.49\textwidth}
            \centering
            \includegraphics[scale=.6]{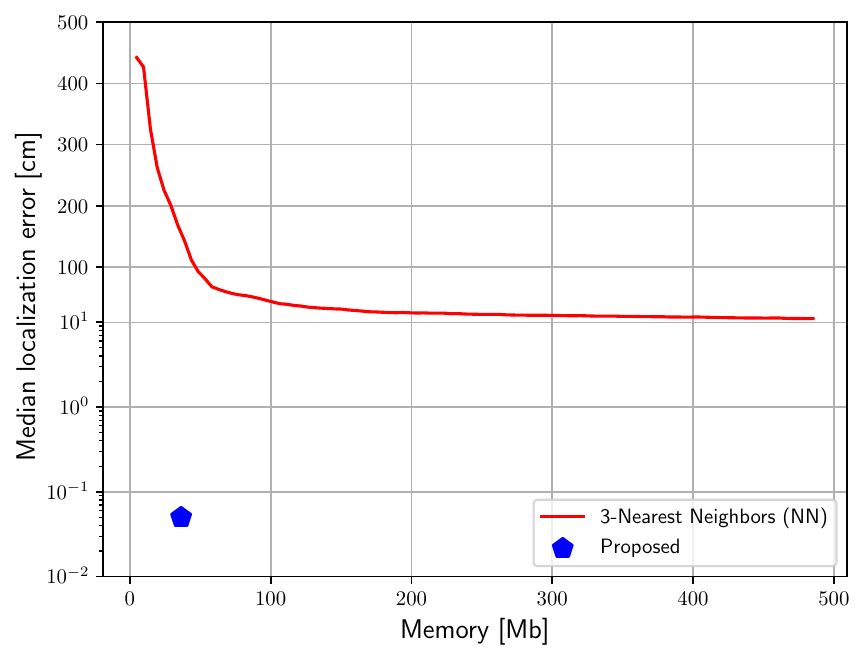}
            \caption{Memory-performance trade-off comparison. The memory requirements of the proposed method originates from the used neural network learnable parameters: further explanations are provided in Section~\ref{sec:experiments}.}
            \label{fig:memory_performance}
        \end{minipage}
    \end{figure*}

    \noindent \textbf{Related work.} ML-based localization methods is a well studied topic in the recent literature: the reader can see~\cite{9264122} and~\cite{burghal2020comprehensivesurveymachinelearning} for comprehensive surveys. Fingerprinting-based localization can rely on a wide range of ML methods: support vector machine~\cite{5779231}, $k$-NN~\cite{8485369}, decision trees~\cite{YIM20081296}, random forests~\cite{9050417}, principal component analysis~\cite{7743586}, but also deep neural networks~\cite{9128640,10188249}.

    On the other hand, the use of the model-based machine learning paradigm has been applied to several wireless communication problems such as precoding~\cite{Lavi23}, detection~\cite{Samuel17}, channel estimation~\cite{Hengtao18,Xiuhong21,yassine2022,Chatelier2022}, angle of arrival estimation~\cite{Shmuel2023,chatelier24_diffMUSIC}, channel charting~\cite{Yassine2022a,yassine2023modelbased,yassine2023charting,chatelier24csi}, and also in integrated sensing and communication scenarios~\cite{mateosramos2023semisupervised,Mateos_Ramos24}. Its use in a localization context has been studied in~\cite{9390409}, where a model-based neural network is employed to estimate the position and velocity of a moving UE. Specifically, these characteristics are obtained by solving a weighted least squares problem, where the weight matrix is dynamically adapted to the input measurements through a neural network. The primary distinction between~\cite{9390409} and the approach proposed in this paper lies in methodology: while~\cite{9390409} leverages data to dynamically solve a conventional least-squares problem, the proposed method exploits prior knowledge of the channel propagation model to perform data augmentation and guide the optimization process. Furthermore, this study considers a static UE and performs localization using only uplink channel coefficients obtained at a single BS, whereas~\cite{9390409} addresses a mobile UE and relies on measurements from multiple BSs processed centrally. Finally, to the best of the author's knowledge, this work is the first to explore the use of a generative neural channel model to perform localization. \Review{While a simplified evaluation setup with noiseless, static channels is considered, this work provides an initial validation of the proposed method, and still reflects relevant real-world scenarios. Indeed, the noiseless assumption is justified by the widespread use of pilot symbols in sensing systems. Additionally, while static or quasi-static propagation environments are not representative of all propagation environments, there exist contexts where this assumption is reasonable, e.g. in certain indoor environments, such as indoor factories \Rours{where precise localization supports automation}. In time-varying environments, periodic data collection combined with re-training of the proposed neural architecture could further extend the approach.}

    \noindent \textbf{Organization.} The rest of the paper is organized as follows. Section~\ref{sec:problem_formulation} presents the system model and defines the localization problem. Section~\ref{sec:proposed_method} presents the proposed location estimator along with theoretical results on its behavior. Section~\ref{sec:experiments} evaluates the proposed method performance and compares it against baselines on realistic synthetic data. Finally, Section~\ref{sec:conclusion} presents some conclusions and directions for future work.

    \noindent \textbf{Notations.} Lowercase bold letters represent vectors while uppercase bold letters represent matrices. $\mathbb{R}$ and $\mathbb{C}$ denote the real and complex fields, while $\mathbb{S}$ and $\mathbb{G}$ denote location sets. $\varnothing$ denotes the empty set. $\espo{\cdot}{\mathbf{x}\sim \mathcal{P}_{\mathbf{x}}}$ denotes the expectancy operator on the random variable $\mathbf{x}$, that follows the distribution $\mathcal{P}_{\mathbf{x}}$. $\norm{\cdot}{p}$ denotes the $\ell_p$ norm, while $\norm{\cdot}{\mathsf{F}}$ denotes the $\ell_{\mathsf{F}}$, Frobenius norm. $\abs{\cdot}$ denotes the absolute value operator for real numbers, modulus operator for complex numbers, and cardinality operator for sets. $\mathcal{C}\left(\mathbf{c},r\right)$ denotes the circle of center $\mathbf{c} \in \mathbb{R}^2$ and radius $r \in \mathbb{R}$. $\times$ denote the Cartesian product set operator. $\cap$ and $\cup$ denote the intersection and union set operators. $\otimes$ and $\odot$ denote the Kronecker and Hadamard matrix products. $\textnormal{diag}\left( \cdot \right)$ denotes the matrix operator constructing a diagonal matrix from a vector and $\textnormal{vec}\left( \cdot \right)$ denotes the vectorization operator. Finally, $\mathcal{U}\left(\mathbb{S}\right)$ denotes the uniform distribution on the subspace defined by $\mathbb{S}$, \Rours{$\mathcal{N}\left(\mu,\sigma^2\right)$ denotes the gaussian distribution of mean $\mu$ and variance $\sigma^2$}, and $\mathds{1}_{\left\{x\in \mathbb{A}\right\}}\left(x\right)$ denotes the indicator function over the set $\mathbb{A}$.

\section{System model and problem formulation}\label{sec:problem_formulation}
    This section defines the system model in addition to the localization problem.
    
    Let us consider a scene where the location space is a plane denoted by $\mathbb{S} \subset \mathbb{R}^2$. The considered task is the localization of a single-antenna UE located at $\mathbf{x} \in \mathbb{S}$, which transmits pilots on $N_s$ distinct frequencies to a BS equipped with $N_a$ antennas, using only the noiseless uplink channel matrix $\mathbf{H}\left(\mathbf{x}\right) \in \mathbb{C}^{N_a \times N_s}$. It is assumed that the considered propagation scene is static.
    \begin{remark}
        \Review{As mentioned in the introduction, the assumptions of noiseless and static channels are physically grounded and reflect practical scenarios, especially indoor environments such as factories, where precise localization is critical for automation.}
    \end{remark}
    
    This localization problem can be formally defined as determining a location estimator $\mathcal{E}^\star: \mathbb{D}_f \rightarrow \mathbb{S}$, that minimizes the mean localization error, i.e.:
    \begin{equation}
        \mathcal{E}^\star = \argmin_{\mathcal{E} \in \mathcal{X}} \espo{\norm{\mathbf{x}-\mathcal{E}\left(f\left(\mathbf{H\left(x\right)}\right)\right)}{2}}{\mathbf{x}\sim \mathcal{P}_{\mathbf{x}}},
    \end{equation}
    where $\mathcal{X}$ is the location estimator class, $f\left(\mathbf{H\left(x\right)}\right)$ is a channel feature, $\mathbb{D}_f$ being its domain, and $\mathcal{P}_{\mathbf{x}}$ denotes the true locations distribution over $\mathbb{S}$. The challenge lies in developing an estimator capable of providing accurate location estimates across the entire scene, with minimal a-priori knowledge. One classical approach in the localization literature is fingerprinting with propagation channel coefficients, which can be summarized as follows:
    \begin{enumerate}
        \item In a given propagation scene, one constitutes a dictionary of known UE locations $\mathcal{G}$:
        \begin{equation}
            \mathcal{G} = \left\{\tilde{\mathbf{x}}_i\right\}_{i=1}^{N_f} \subset \mathbb{S},
        \end{equation}
        and a dictionary of the corresponding propagation channel coefficients $\mathcal{H}$:
        \begin{equation}
            \mathcal{H} = \left\{\mathbf{H}\left(\tilde{\mathbf{x}}_i\right) \right\}_{i=1}^{N_f},
        \end{equation}
        where $\mathbf{H}\left(\tilde{\mathbf{x}}_i\right) \in \mathbb{C}^{N_a \times N_s}$ is the uplink channel matrix between the BS and UE located at $\tilde{\mathbf{x}}_i \in \mathbb{S}$.
        \item For the UE with unknown location $\mathbf{x}$, the BS estimates the uplink channel matrix $\mathbf{H}\left(\mathbf{x}\right)$. The estimated location is then obtained as:
        \begin{equation}\label{eq:fingerprinting_pb}
            \hat{\mathbf{x}}\left(\mathbf{H}\left(\mathbf{x}\right)\right) = \argmax_{\tilde{\mathbf{x}} \in \mathcal{G}} \simil{\mathbf{H}\left(\mathbf{x}\right),\mathbf{H}\left(\tilde{\mathbf{x}}\right)},
        \end{equation}
        with $\simil{\cdot,\cdot} \in \left]-\infty,1\right]$ being a channel similarity measure, optimal at $1$.
    \end{enumerate}
    This approach exhibits two shortcomings. Firstly, in order for the estimated location to be unique within $\mathcal{G}$, the channel function needs to be injective with respect to the similarity measure, as defined below.
    \begin{definition}\label{def:exact_injectivity}
        The channel function is said to be injective wrt. similarity measure iff, $\forall \left(\mathbf{x}_m,\mathbf{x}_n\right) \in \mathbb{S} \times \mathbb{S}$: 
        \begin{equation}
            \simil{\mathbf{H}\left(\mathbf{x}_m\right), \mathbf{H}\left(\mathbf{x}_n\right)} = 1 \Rightarrow \mathbf{x}_m = \mathbf{x}_n.
        \end{equation}
    \end{definition}

    Secondly, the precision of the estimated location is directly related to the spatial density of the location subset $\mathcal{G}$. Indeed, considering a uniform true locations distribution, a manifestation of the well-known curse of dimensionality~\cite[Chapter 2, pp.22-26]{ESL} reads:
    \begin{equation}\label{eq:fingerprinting_resolution}
        \espo{\norm{\mathbf{x}-\hat{\mathbf{x}}\left(\mathbf{H}\left(\mathbf{x}\right)\right)}{2}}{\mathbf{x}\sim \mathcal{P}_{\mathbf{x}}} \propto \left(\frac{A_s}{\abs{\mathcal{G}}}\right)^{\frac{1}{d}},
    \end{equation}
    where $\hat{\mathbf{x}}\left(\mathbf{H}\left(\mathbf{x}\right)\right)$ is obtained by solving Eq.~\eqref{eq:fingerprinting_pb}, $A_s$ is the area covered by $\mathbb{S}$, and $d=\dim_{\mathbb{R}}\left(\mathbb{S}\right) = 2$. Enhancing the localization accuracy thus necessitates an extensive measurement process, which poses challenges in both the time required for dictionary construction, and the memory requirements for storing the dictionary. \Review{For a fixed covered area $A_s$, dividing the localization error by a factor $n \in \mathbb{R}^*$ requires multiplying the cardinality of $\mathcal{G}$ by a factor $n^d$. This results in an impractically large dataset size, which represents the main limitation of classical fingerprinting methods.} The use of the model-based machine learning paradigm to overcome the limitations of classical fingerprinting methods is presented in the next section.

\section{Proposed method}\label{sec:proposed_method}
    This section presents how a model-based neural network learning the location-to-channel mapping can be used to improve the classical fingerprinting-based localization method.

    \subsection{Learning a generative neural channel model}\label{subsec:generative}
    It is proposed to use the neural architecture presented in~\cite{chatelier_loc2chan_icassp24}, and expanded in~\cite{chatelier_loc2chan24}, as a generative neural channel model. Considering $L_p$ propagation paths, the channel coefficients at frequency $f_k$ between antenna $j$ of the BS and the UE located at $\mathbf{x} \in \mathbb{S}$ can be modeled as:
    \begin{equation}\label{eq:channel_model}
        h_j \left(f_k,\mathbf{x}\right) = \sum_{l=1}^{L_p} \dfrac{\gamma_l}{d_{l,j}\left(\mathbf{x}\right)} \ej[-]{\frac{2\pi}{\lambda_k} d_{l,j}\left(\mathbf{x}\right)},
    \end{equation}
    where $\gamma_l \in \mathbb{C}$ accounts for wave-matter interactions of NLoS paths, $d_{l,j}\left(\mathbf{x}\right)$ is the length of the $l$th propagation path for the $j$th BS antenna, and $\lambda_k \triangleq c/f_k$~\footnote{\scriptsize{Eq.~\eqref{eq:channel_model} accounts only for direct propagation, reflection, and diffraction. Note that scattering effects can be incorporated by decoupling the distances used in the attenuation term (as a product of distances) and in the phase term (as a sum of distances).}}. Let $\mathbf{H}\left(\mathbf{x}\right) \in \mathbb{C}^{N_a \times N_s}$ be the uplink channel matrix at location $\mathbf{x}$, i.e. the concatenation of the channel coefficients, \Rours{presented in Remark.~\ref{remark:virtual_sources}}, for all antennas and frequencies. The neural network $\ftheta$, with $\boldsymbol{\theta}$ its set of learnable parameters, learns the mapping:
    \begin{equation}
        \begin{aligned}
            \ftheta\colon \mathbb{R}^2 &\longrightarrow \mathbb{C}^{N_a \times N_s}\\
            \mathbf{x} &\longrightarrow \mathbf{H}\left(\mathbf{x}\right),
        \end{aligned}
    \end{equation}
    by minimizing the loss function:
    \begin{equation}\label{eq:loss}
        \mathcal{J} = \frac{1}{\abs{\mathcal{D}}}\sum_{\left(\mathbf{x},\mathbf{H}\left(\mathbf{x}\right)\right)\in\mathcal{D}} \norm{\mathbf{H}\left(\mathbf{x}\right)-\ftheta\left(\mathbf{x}\right)}{\mathsf{F}}^2,
    \end{equation}
    over the training dataset $\mathcal{D}$ consisting of locations and associated uplink channel matrices.

    \begin{remark}
        In~\cite{chatelier_loc2chan24}, all theoretical developments are made for true locations $\mathbf{x} \in \mathbb{R}^3$. This study only considers the $\mathbb{R}^2$ scenario to simplify the associated theoretical analysis. This assumption is equivalent to having all UEs confined to the same elevation plane as the BS. All the proposed theoretical results can easily be extended to $\mathbb{R}^3$.
    \end{remark}

    \begin{remark}\label{remark:virtual_sources}
        Using the virtual source theory~\cite[Chapter 1, pp.47-49]{pozar2011microwave},~\cite{Deliang2014}, one can rewrite Eq.~\eqref{eq:channel_model} as:
        \begin{equation}\label{eq:channel_model_virtual_sources}
            h_j \left(f_k,\mathbf{x}\right) = \sum_{l=1}^{L_p} \dfrac{\gamma_l}{\norm{\mathbf{x}-\mathbf{a}_{l,j}}{2}} \ej[-]{\frac{2\pi}{\lambda_k} \norm{\mathbf{x}-\mathbf{a}_{l,j}}{2}},
        \end{equation}
        where $\mathbf{a}_{l,j} \in \mathbb{R}^2$ is the $j$th true antenna location, when $l = 1$, or $j$th virtual antenna location, when $l>1$. This channel model will be used in the subsequent theoretical developments.
    \end{remark}

    \begin{figure}[!t]
        \centering
        \includegraphics[width=.95\columnwidth]{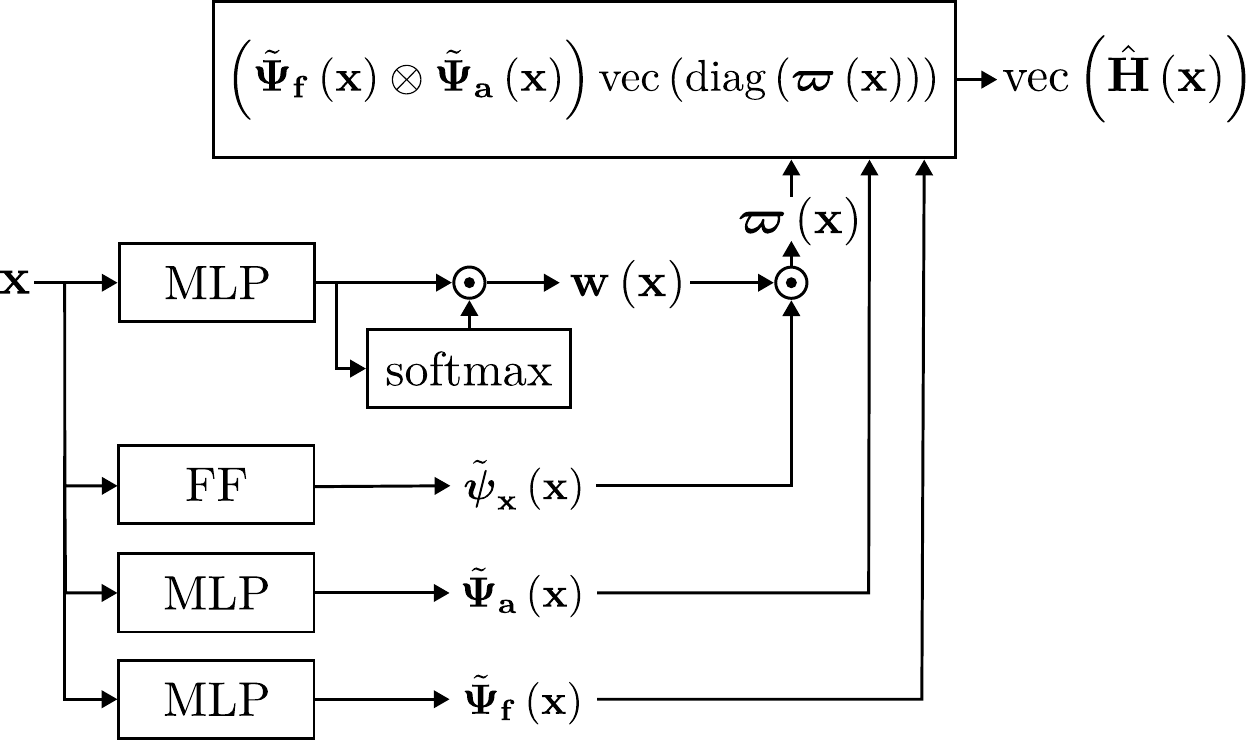}
        \caption{Architecture of the neural model $\ftheta$ proposed in~\cite{chatelier_loc2chan24}. The Fourier Feature (FF) embedding projects the input location $\mathbf{x}$ into a subspace that captures high frequency variations.}
        \label{fig:neural_archi}
    \end{figure}

    The architecture of the network proposed in~\cite{chatelier_loc2chan24} is recalled in Fig.~\ref{fig:neural_archi}. For a given input location $\mathbf{x}$, the associated channel $\mathbf{H}\left(\mathbf{x}\right)$ is infered as a sparse linear combination of steering vectors (denoted by the dictionary $\tilde{\mathbf{\Psi}}_{\mathbf{a}}\left(\mathbf{x}\right)$), of frequency response vectors (denoted by the dictionary $\tilde{\mathbf{\Psi}}_{\mathbf{f}}\left(\mathbf{x}\right)$) and planar wavefronts (denoted by the dictionary $\tilde{\boldsymbol{\psi}}_{\mathbf{x}}\left(\mathbf{x}\right)$), which accounts for local displacements. \Review{Each MLP consists of a distinct three-layer MLP hypernetwork, with hidden layer sizes of $256$ for the weight hypernetwork and $64$ for the antenna and frequency hypernetworks. In total, the model comprises about $9.1$ million learnable parameters.} It was shown that, after training, this neural architecture was able to infer the channel matrix at any wanted location in the scene considered during training, with minimal error, yielding the following definition.
    \begin{definition}
        \label{def:ftheta_error}
        Let $\mathbb{S} \subset \mathbb{R}^2$ be the set of locations covered by the considered scene, and $\ftheta$ the trained neural model. Then, the model estimation error $\Xi\left(\mathbf{x}\right)$ is bounded as, $\exists \epsilon_{\ftheta} \in \mathbb{R}^+$ s.t. $\forall \mathbf{x} \in \mathbb{S}:$
        \begin{equation}
            \Xi\left(\mathbf{x}\right) \triangleq \dfrac{\norm{\mathbf{H}\left(\mathbf{x}\right)-\ftheta\left(\mathbf{x}\right)}{\mathsf{F}}}{\norm{\mathbf{H}\left(\mathbf{x}\right)}{\mathsf{F}}} \leq \epsilon_{\ftheta}.
        \end{equation}
    \end{definition}

    The estimation error $\Xi\left(\mathbf{x}\right)$ of the proposed network remains low, even in complex radio environments. In~\cite{chatelier_loc2chan24}, the normalized mean squared error (NMSE), defined as $\espo{\Xi\left(\mathbf{x}\right)^2}{\mathbf{x} \sim \mathcal{P}_{\mathbf{x}}}$, was reported to be approximately $-40$ dB in LoS-only environments, $-29$ dB in LoS/NLoS environments and $-20$ dB in NLoS-only environments. \Redit{This good performance on the test dataset indicates that the proposed network learned to interpolate between training locations. This capability is enabled by architectural constraints that originate from a physical propagation channel model, as presented in details in~\cite{chatelier_loc2chan24}.} Furthermore, the proposed network was shown to effectively capture the fast-fading characteristics of the channel, i.e. small-scale fading, allowing it to infer distinct channel matrices for locations separated by sub-wavelength distances, while maintaining low estimation error. These findings support the feasibility of using $\ftheta$ as a generative neural channel model within the considered scene.

    \begin{remark}
        Once the location-to-channel mapping is learned, it can be used to generate any number of approximated channel at any given location in the scene, allowing to generate a potentially infinite database to be used for localization. Moreover, the mapping being continuous and differentiable, first order optimization method can be applied. These two properties are at the core of the proposed method, which is described in details next.
    \end{remark}

    \subsection{Application to localization}\label{subsec:localization_application}
    As presented in Eq.~\eqref{eq:fingerprinting_resolution} in Section~\ref{sec:problem_formulation}, the primary drawbacks of fingerprinting-based localization methods include the need for a substantial comparison dictionary, which in turn necessitates an extensive measurement campaign and imposes memory storage constraints. It is proposed to use the proposed model-based neural architecture to solve this dictionary size issue. Indeed, as $\ftheta$ can be seen as a generative neural channel model, an infinite number of comparison channels can be generated on the fly to solve Eq.~\eqref{eq:fingerprinting_pb},~\footnote{Note that, when comparing the approaches in Section~\ref{sec:experiments}, the amount of data required to train the generative neural model is taken into account.} as formalized in Definition~\ref{def:estimator}.
    \begin{mdframed}
        \begin{definition}\label{def:estimator}
            Let $\mathbb{S} \subset \mathbb{R}^2$ be the set of locations covered by the considered scene. Let $\mathbf{x} \in \mathbb{S}$ be the UE unknown location and $\mathbf{H}\left(\mathbf{x}\right)\in \mathbb{C}^{N_a \times N_s}$ be the associated uplink channel matrix. An estimate of the UE location is given by:
            \begin{align}
                \hat{\mathbf{x}}\left(\mathbf{H}\left(\mathbf{x}\right) \vert \boldsymbol{\theta}, \mathbb{S}\right) &\triangleq \argmin_{\tilde{\mathbf{x}} \in \mathbb{S}} \norm{\mathbf{H}\left(\mathbf{x}\right) - \ftheta\left(\tilde{\mathbf{x}}\right)}{\mathsf{F}}.\label{eq:mb_fingerprinting}
            \end{align}
        \end{definition}
    \end{mdframed}

    \begin{remark}
        Eq.~\eqref{eq:mb_fingerprinting} is obtained by considering the following similarity measure in Eq.~\eqref{eq:fingerprinting_pb}:
        \begin{equation}\label{eq:sim_measure}
            \simil{\mathbf{H}\left(\mathbf{x}_m\right), \mathbf{H}\left(\mathbf{x}_n\right)} = 1 - \norm{\mathbf{H}\left(\mathbf{x}_m\right)- \mathbf{H}\left(\mathbf{x}_n\right)}{\mathsf{F}}.
        \end{equation}
        This similarity measure is optimal when one-valued, as it requires the channel coefficients to be exactly equal across all antennas and subcarriers. Its injectivity properties and their implications will be discussed in Section~\ref{subsec:injectivity}. Since the Frobenius distance component of this similarity measure is central to the subsequent analysis, it is referred to as the phase-sensitive (PS) distance and defined as follows:
        \begin{equation}\label{eq:frob_sim_component}
            \mathcal{L}_{\mathsf{PS}}\left(\mathbf{H}\left(\mathbf{x}\right),\tilde{\mathbf{x}} \vert \boldsymbol{\theta}\right) \triangleq \norm{\mathbf{H}\left(\mathbf{x}\right) - \ftheta\left(\tilde{\mathbf{x}}\right)}{\mathsf{F}}.
        \end{equation}
    \end{remark}
    
    Given the ability of $\ftheta$ to effectively learn the location-to-channel mapping, the theoretical error performance of the proposed method can be established, as formalized in Theorem~\ref{thm:min_error}.
        \begin{mdframed}
        \begin{theorem}\label{thm:min_error}
            Let $\mathbb{S} \subset \mathbb{R}^2$ be the set of locations covered by the considered scene. Let $\mathbf{x} \in \mathbb{S}$ be the UE unknown location. Let $\hat{\mathbf{x}}\left(\mathbf{H}\left(\mathbf{x}\right) \vert \boldsymbol{\theta}, \mathbb{S}\right)$ be the solution of Eq.~\eqref{eq:mb_fingerprinting}. Let $\epsilon_{\ftheta} \in \mathbb{R}^+$ be the model estimation error bound defined in Definition~\ref{def:ftheta_error}. Assuming that Eq.~\eqref{eq:sim_measure} satisfies Definition~\ref{def:exact_injectivity}, it follows that:
            \begin{equation}
                \sup_{\mathbf{x} \in \mathbb{S}}\norm{\mathbf{x}-\hat{\mathbf{x}}\left(\mathbf{H}\left(\mathbf{x}\right) \vert \boldsymbol{\theta}, \mathbb{S}\right)}{2}\mathop{\longrightarrow}_{\epsilon_{\ftheta} \rightarrow 0} 0.
            \end{equation}
        \end{theorem}
    \end{mdframed}
    \begin{proof}
        See Appendix~\ref{appendices:thm_min_error}.
    \end{proof}

    Theorem~\ref{thm:min_error} demonstrates that the location estimator presented in Definition~\ref{def:estimator} achieves theoretically perfect accuracy while maintaining finite memory requirements. Indeed, channel coefficients can be evaluated on-the-fly using $\ftheta$, which only requires storing $\abs{\boldsymbol{\theta}}$ coefficients. However, the proposed method requires inferring the channel at every location in $\mathbb{S}$, which is uncountable. To address this, the scene location space is discretized using a finite grid $\mathbb{G} \subset \mathbb{S} \subset \mathbb{R}^2$. This discretization leads to the following localization error bound.
    \begin{mdframed}
        \begin{theorem}\label{thm:lambda_bound}
            Let $\mathbb{G} \subset \mathbb{S} \subset \mathbb{R}^2$ be the discretized location grid. Let $\mathbf{x} \in \mathbb{S}$ be the UE unknown location. Assuming that the true locations are uniformly distributed within $\mathbb{S}$, and that the grid is uniform with inter-element spacing $\nu \in \mathbb{R}^{+,*}$ yields:
            \begin{equation}
                \espo{\norm{\mathbf{x}-\hat{\mathbf{x}}\left(\mathbf{H}\left(\mathbf{x}\right) \vert \boldsymbol{\theta}, \mathbb{G} \right)}{2}}{\mathbf{x}\sim \mathcal{P}_{\mathbf{x}}} \geq \nu \delta,
            \end{equation}
            with $\delta = \frac{1}{6} \left(\sqrt{2}+\ln\left(\sqrt{2}+1\right)\right) \simeq 0.38$.
        \end{theorem}
    \end{mdframed}
    \begin{proof}
        See Appendix~\ref{appendix:a}.
    \end{proof}

    \begin{remark}
        Theorem~\ref{thm:lambda_bound} states that the localization accuracy is lower-bounded by the grid inter-element spacing. This is unsurprising as solving Eq.~\eqref{eq:mb_fingerprinting} on the grid $\mathbb{G}$ is equivalent to performing a grid search, inherently constraining the localization accuracy to the grid resolution.
    \end{remark}

    To overcome this grid limitation, an off-grid approach is introduced: after the initial grid search, a gradient-based refinement with $N_{\nabla}$ steps is applied to enhance localization accuracy, as formalized in Eq.~\eqref{eq:gd_approach}. At step $k$, the update is expressed as:
    \begin{align}\label{eq:gd_approach}
        \hat{\mathbf{x}}^{\left(k+1\right)} \leftarrow  &\hat{\mathbf{x}}^{\left(k\right)} - \alpha \nabla_{\tilde{\mathbf{x}}} \mathcal{L}_{\mathsf{PS}}\left(\mathbf{H}\left(\mathbf{x}\right),\tilde{\mathbf{x}} \vert \boldsymbol{\theta}\right) \vert_{\hat{\mathbf{x}}^{\left(k\right)}},
    \end{align}
    where $\alpha \in \mathbb{R}^+$ controls the gradient contribution to the update, $\hat{\mathbf{x}}^{\left(k\right)}$ is the location estimate at step $k$, with $\hat{\mathbf{x}}^{\left(0\right)} = \hat{\mathbf{x}}\left(\mathbf{H}\left(\mathbf{x}\right) \vert \boldsymbol{\theta}, \mathbb{G} \right)$. As for all gradient-based methods, this localization refinement is prone to local minima error in the PS distance used as a loss function in Eq.~\eqref{eq:gd_approach}. A method to avoid local minima and increase performance is presented in Section~\ref{subsec:perf_optim}.

    \subsection{About the channel function injectivity}\label{subsec:injectivity}

    \begin{remark}
        \Redit{As discussed in Section~\ref{sec:problem_formulation}, ensuring that the channel function is injective wrt. the similarity measure defined in Eq.~\eqref{eq:sim_measure} is essential for the uniqueness of the location estimate. This can be demonstrated by deriving the global minimum of the PS distance.}
    \end{remark}
    The subsequent theoretical contributions can be summarized as follows:
    \begin{itemize}
        \item Theorem~\ref{thm:minima_frob} establishes a condition for reaching the global minimum of the Frobenius distance introduced in Eq.~\eqref{eq:frob_sim_component}.
        \item Corollary~\ref{corol:approx_injectivity} demonstrates that the channel function is injective wrt. the similarity measure defined in Eq.~\eqref{eq:sim_measure}, based on the results of Theorem~\ref{thm:minima_frob}.
        \item Corollary~\ref{corol:minima_distance} builds on Theorem~\ref{thm:minima_frob}, considering a dominant propagation path to define the distance between two consecutive minima in the Frobenius distance.
    \end{itemize}

    \begin{mdframed}
        \begin{theorem}\label{thm:minima_frob}
            For a given location $\mathbf{x} \in \mathbb{S}$, the Frobenius distance component of the similarity measure defined in Eq.~\eqref{eq:frob_sim_component}, rewritten as $\norm{\mathbf{H}\left(\mathbf{x}\right)-\mathbf{H}\left(\mathbf{x}+\bdelta\right)}{\mathsf{F}}$, reaches a global minimum when the displacement vector $\bdelta \in \mathbb{R}^2$ satisfies:
            \begin{equation}
                \bdelta \in \bigcap_{\left(i,j,l,k\right) \in \mathbb{L}} \mathcal{C}\left(\mathbf{a}_{l,i}-\mathbf{x},\norm{\mathbf{x} - \mathbf{a}_{l,i}}{2}+k \lambda_j\right). \label{eq:thm_condition_global_min}
            \end{equation}
            with $\mathbb{L} = \llbracket 1,N_a\rrbracket \times \llbracket 1,N_s\rrbracket \times \llbracket 1,L_p\rrbracket \times \mathbb{Z}$, and $\mathbf{a}_{l,i} \in \mathbb{R}^2$ the antenna locations defined in Remark~\ref{remark:virtual_sources}.
        \end{theorem}
    \end{mdframed}
    \begin{proof}
        See Appendix~\ref{appendix:c}.
    \end{proof}

    \begin{corollary}\label{corol:approx_injectivity}
        In the absence of symmetries in the scene including the antenna array, and considering $N_a > 1$ antennas, the channel function is injective wrt. the similarity measure defined in Eq.~\eqref{eq:sim_measure}.
    \end{corollary}
    \begin{proof}
        See Appendix~\ref{appendix:d}.
    \end{proof}

    \begin{figure}[t]
        \centering
        \includegraphics[width=.95\columnwidth]{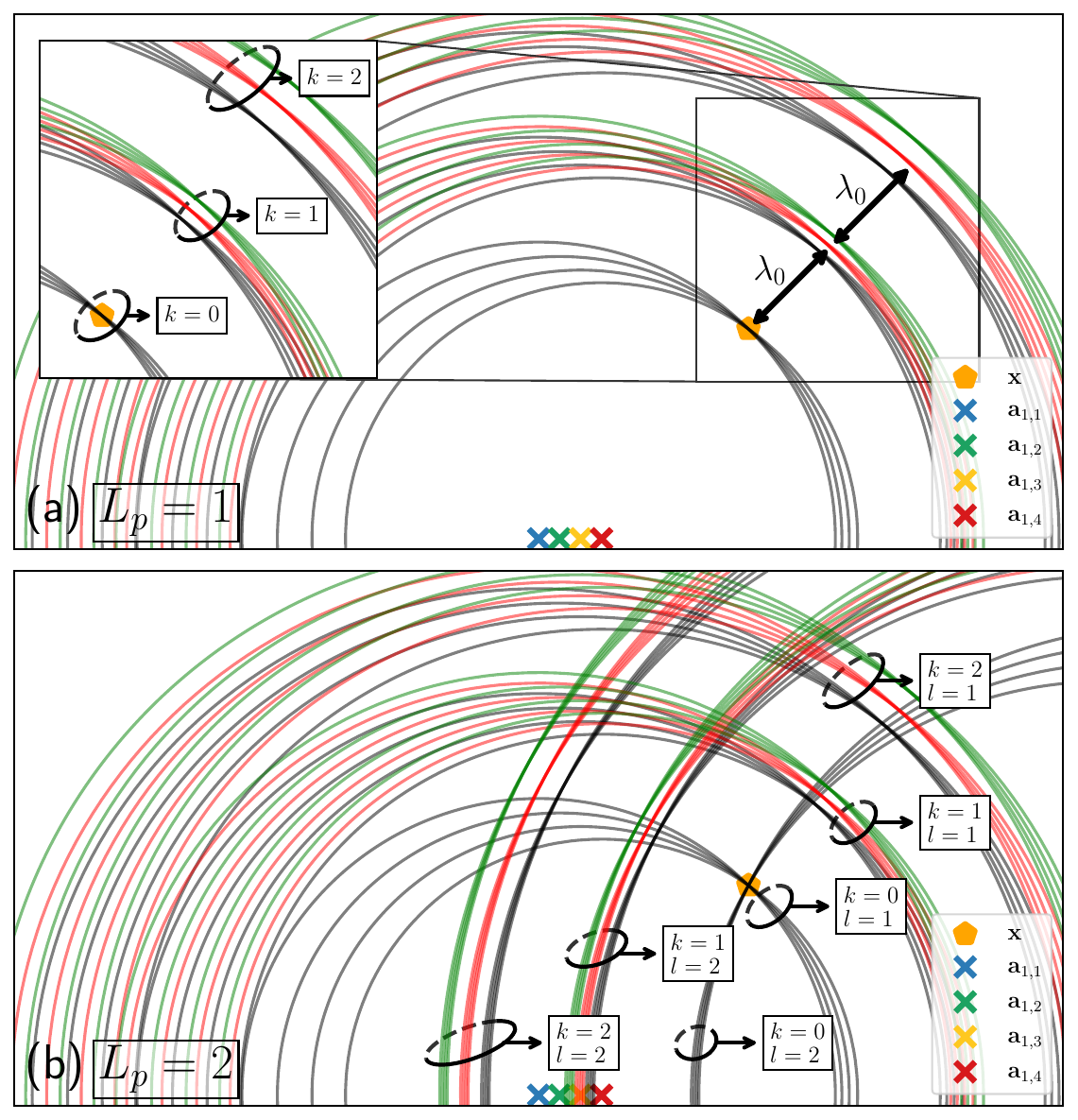}
        \caption{Illustration of Theorem~\ref{thm:minima_frob} in a vector space whose origin is $\mathbf{x}$: minimum circles for $N_a = 4$ antennas, $N_s = 3$ frequencies (color-coded as gray/red/green), and $k \in \llbracket 0, 2\rrbracket$. For $k=0$, circles originating from all frequencies coincide, as observed in Eq~\eqref{eq:thm_condition_global_min}.}
        \label{fig:circle_ambiguity}
    \end{figure}


    Fig.~\ref{fig:circle_ambiguity} illustrates Theorem~\ref{thm:minima_frob} and Corollary~\ref{corol:approx_injectivity} by depicting, for each antenna, path, frequency, and integer multiple $k$, the circles defined in Eq.~\eqref{eq:thm_condition_global_min}. These circles represent local minima, as a displacement vector along one of them nulls the contribution of a specific antenna, frequency, \Rours{and path in $\norm{\mathbf{H}\left(\mathbf{x}\right)-\mathbf{H}\left(\mathbf{x}+\bdelta\right)}{\mathsf{F}}$.} 
    \begin{remark}\label{remark:injectivity}
        As stated by Theorem~\ref{thm:minima_frob}, and visualized in Fig.~\ref{fig:circle_ambiguity}, the global minimum is attained when all circles intersect, which occurs at $\mathbf{x}$, only when $k=0$: this is equivalent to setting $\bdelta = 0_{\mathbb{R}^2}$. This holds only if the scene does not exhibit perfect topological symmetries, when including the antenna array. However, this assumption is weak, as perfect symmetries, up to the scene objects' electromagnetic properties, is unlikely to be found in real-world propagation environments.
    \end{remark}
    The size of the basins of attraction associated to minima in the Frobenius distance is influenced by the number of antennas and the considered bandwidth. As illustrated in Fig.~\ref{fig:circle_ambiguity} \textsf{(a)}, \Rours{a larger antenna array decreases the tangential extent of the region where the circles intersect. Conversely, the radial size of this region depends on $k\left(\lambda_{N_s}-\lambda_1\right), k \in \mathbb{Z}$, which is related to the bandwidth. Theoretically, for a sufficiently high $k$, or a sufficiently high bandwidth, it is possible to resolve one basin of attraction per frequency}. However, in practice, since the wavelength does not significantly vary over the system bandwidth, this distinction remains impractical. This insight clarifies the distance between consecutive minima, as formalized in Corollary~\ref{corol:minima_distance}.
    \begin{mdframed}
        \begin{corollary}\label{corol:minima_distance}
        For a given location $\mathbf{x} \in \mathbb{S}$, in presence of a dominant propagation path, the distance between two consecutive minima in $\norm{\mathbf{H}\left(\mathbf{x}\right)-\mathbf{H}\left(\mathbf{x}+\bdelta\right)}{\mathsf{F}}$ can be approximated as $\lambda_0$, where $\lambda_0 \triangleq 1/N_s \sum_{i\in\llbracket 1, N_s\rrbracket}\lambda_i$. 
        \end{corollary}
    \end{mdframed}
    \begin{proof}
        See Appendix~\ref{appendix:e}, and Fig.~\ref{fig:circle_ambiguity} \textsf{(a)} for an illustration.
    \end{proof}

    \begin{remark}
        When considering multipath propagation, each propagation path introduces an additional set of circles. In this case, the local minima form a grid-like pattern, as exposed in Fig.~\ref{fig:circle_ambiguity} \textsf{(b)}. It increases the complexity of the gradient descent procedure, making convergence to the global minimum more challenging. This is consistent with the well-established fact that localization in the presence of NLoS paths is a hard task.
    \end{remark}

    \begin{algorithm}[t]
        \caption{Proposed enhanced localization method: Off-grid (PI/PS).}
        \label{alg:prop_approach}
        \begin{algorithmic}[1]
            \Require Measured channel matrix $\mathbf{H}\left(\mathbf{x}\right) \in \mathbb{C}^{N_a \times N_s}$, trained neural model $\ftheta$, selected initialization loss $\mathcal{L}_{\mathsf{init}}$, global location grid $\Gt$, local grid size $L$, inter-element spacing $\nu$, number of circles $N_{\mathsf{C}}$, number of circle points $M_{\mathsf{C}}$.
            \State \textit{($\mathrm{a}$ in Fig.~\ref{fig:circle_explanation})}. Compute the local grid initialization location by performing a grid search on $\Gt$: $\tilde{\mathbf{x}}_{\mathrm{i}} \leftarrow \hat{\mathbf{x}}\left(\mathbf{H}\left(\mathbf{x}\right) \vert \boldsymbol{\theta}, \Gt\right)$, with:
            \begin{equation}\label{eq:init_grid_search}
                \hat{\mathbf{x}}\left(\mathbf{H}\left(\mathbf{x}\right) \vert \boldsymbol{\theta}, \Gt\right) = \argmin_{\tilde{\mathbf{x}} \in \Gt} \mathcal{L}_{\mathsf{init}}\left(\mathbf{H}\left(\mathbf{x}\right),\tilde{\mathbf{x}}\vert \boldsymbol{\theta}\right).
            \end{equation}
            \State Construct the local grid $\Gl$: size $L$ by $L$, centered at $\tilde{\mathbf{x}}_{\mathrm{i}}$, with inter-element spacing $\nu$ in both directions:
            \begin{equation}
                \Gl = \left(\tilde{x}_{\mathrm{i},0} + \mathbb{Z}_{\nu,L}\right) \times \left(\tilde{x}_{\mathrm{i},1} + \mathbb{Z}_{\nu,L}\right),
            \end{equation}
            with $\mathbb{Z}_{\nu,L} = \nu \mathbb{Z} \cap \left[-L/2,L/2 \right]$.
            \State \textit{($\mathrm{b}$ in Fig.~\ref{fig:circle_explanation})}. Perform a grid search on $\Gl$: $\tilde{\mathbf{x}}_{\mathrm{g}} \leftarrow \hat{\mathbf{x}}\left(\mathbf{H}\left(\mathbf{x}\right) \vert \boldsymbol{\theta}, \Gl\right)$, with:
            \begin{equation}
                \hat{\mathbf{x}}\left(\mathbf{H}\left(\mathbf{x}\right) \vert \boldsymbol{\theta}, \Gl\right) = \argmin_{\tilde{\mathbf{x}} \in \Gl} \mathcal{L}_{\mathsf{PS}}\left(\mathbf{H}\left(\mathbf{x}\right),\tilde{\mathbf{x}}\vert \boldsymbol{\theta}\right).
            \end{equation}
            \State \textit{($\mathrm{c}$ in Fig.~\ref{fig:circle_explanation})}. Perform $N_{\nabla}$ gradient steps from $\tilde{\mathbf{x}}_{\mathrm{g}}$, as defined in Eq.~\eqref{eq:gd_approach}, to obtain $\tilde{\mathbf{x}}_{\mathrm{gd}}$.
            \State Evaluate the Frobenius distance at the obtained location:
            \begin{equation}
                \gamma_{\mathrm{gd}} = \mathcal{L}_{\mathsf{PS}}\left(\mathbf{H}\left(\mathbf{x}\right),\tilde{\mathbf{x}}_{\mathrm{gd}} \vert \boldsymbol{\theta}\right).
            \end{equation}
            \State Sample locations on the circles centered at $\tilde{\mathbf{x}}_{\mathrm{gd}}$, with radii $k \lambda_0, k \in \llbracket 1,N_{\mathsf{C}} \rrbracket$:
            \begin{equation}
                \Gc \subset \bigcup_{k \in \llbracket 1,N_{\mathsf{C}} \rrbracket} \mathcal{C}\left(\tilde{\mathbf{x}}_{\mathrm{gd}}, k \lambda_0\right), \abs{\Gc} = M_{\mathsf{C}}.
            \end{equation}
            \State \textit{($\mathrm{d}$ in Fig.~\ref{fig:circle_explanation})}. Find the minimum Frobenius distance on the circles:
            \begin{align}
                \tilde{\mathbf{x}}_{\mathrm{c}}^{\star} &= \argmin\limits_{\tilde{\mathbf{x}}_{\mathrm{c}} \in \Gc} \mathcal{L}_{\mathsf{PS}}\left(\mathbf{H}\left(\mathbf{x}\right),\tilde{\mathbf{x}}_{\mathrm{c}} \vert \boldsymbol{\theta}\right)\\
                \gamma_{\mathrm{c}}^{\star} &= \mathcal{L}_{\mathsf{PS}}\left(\mathbf{H}\left(\mathbf{x}\right),\tilde{\mathbf{x}}_{\mathrm{c}}^{\star} \vert \boldsymbol{\theta}\right).
            \end{align}
            \Statex Compare the Frobenius distances:
            \If{$\gamma_{\mathrm{c}}^{\star} \leq \gamma_{\mathrm{gd}}$}
                \State \textit{($\mathrm{e}$ in Fig.~\ref{fig:circle_explanation})}. Perform $N_{\nabla}$ gradient steps from $\tilde{\mathbf{x}}_{\mathrm{c}}^{\star}$, as defined in Eq.~\eqref{eq:gd_approach}, to obtain $\tilde{\mathbf{x}}_{\mathrm{gd_2}}$.
                \State Update the estimated location: $\hat{\mathbf{x}} \leftarrow \tilde{\mathbf{x}}_{\mathrm{gd_2}}$.
            \Else 
                \State Update the estimated location: $\hat{\mathbf{x}} \leftarrow \tilde{\mathbf{x}}_{\mathrm{gd}}$.
            \EndIf
            \Ensure Estimated location: $\hat{\mathbf{x}}$.
        \end{algorithmic}
    \end{algorithm}

    \subsection{Performance and complexity optimization}\label{subsec:perf_optim}

    Hitherto, Theorem~\ref{thm:lambda_bound} has established the theoretical limitations of the grid-based approach, Eq.~\eqref{eq:gd_approach} has introduced a gradient-based refinement, and Corollary~\ref{corol:approx_injectivity} has demonstrated the suitability of the proposed similarity measure for this localization task. It is now proposed to optimize the performance of the proposed method, while reducing its computational complexity. For a given true location, the accuracy of the proposed method depends on two key assumptions. Firstly, the location grid $\Gn$ must include a point sufficiently close to the true location. Secondly, the gradient descent process must successfully converge to the global minimum. To satisfy these assumptions, the cardinality of $\Gn$ needs to be high, which entails significant computational complexity. Additionally, even under this condition, the effectiveness of gradient descent in reaching the global minimum remains uncertain. Indeed, as established in Corollary~\ref{corol:minima_distance}, the minima of the PS distance are separated by a distance approximately equal to the central wavelength, which, in conventional communication systems, typically ranges from centimeters to millimeters. This small spacing between local minima underscores the risk of convergence to a local minimum when conducting the gradient-based approach. The proposed enhanced localization method \Redit{is presented in Algorithm~\ref{alg:prop_approach}~\footnote{\Redit{The phase-insensitive (PI) term appearing in Algorithm~\ref{alg:prop_approach} is thoroughly described in the subsequent \textit{Complexity optimization} paragraph.}}}; details about performance and complexity optimization are subsequently presented, and Fig.~\ref{fig:circle_explanation} illustrates the proposed method and variants. Note that the \textit{naive} approaches use the large discretized location grid $\Gn$.

    \noindent\textbf{Performance optimization.} The theoretical knowledge of the distance between consecutive minima is incorporated into the gradient-based refinement procedure. Specifically, a dual-gradient-descent scheme is considered: the first gradient descent is performed from $\tilde{\mathbf{x}}_{\mathrm{g}}$, i.e. the location obtained through the on-grid approach. This approach allows to obtain the location estimate $\tilde{\mathbf{x}}_{\mathrm{gd}}$. Subsequently, $M_{\mathsf{C}}$ locations are sampled on $N_{\mathsf{C}}$ circles of center $\tilde{\mathbf{x}}_{\mathrm{gd}}$ and radii corresponding to integer multiples of the theoretical separation distance between minima, i.e. $k \lambda_0, k \in \mathbb{Z}$ (line $6$ in Algorithm~\ref{alg:prop_approach}). The minimum Frobenius distance on the sampled circle locations, i.e. $\gamma_{\mathrm{c}}^{\star}$, associated to location $\tilde{\mathbf{x}}_{\mathrm{c}}^{\star}$, is then compared to the value obtained at the location estimated through the initial gradient-descent procedure. If it is lower, it indicates the presence of a better minimum, suggesting the possibility of refining the location estimate towards the global minimum. In that case, a second gradient-descent procedure is performed from $\tilde{\mathbf{x}}_{\mathrm{c}}^{\star}$ to ensure convergence to the lowest point within the basin of attraction.

    \begin{figure*}[t]
        \centering
        \includegraphics[width=.95\textwidth]{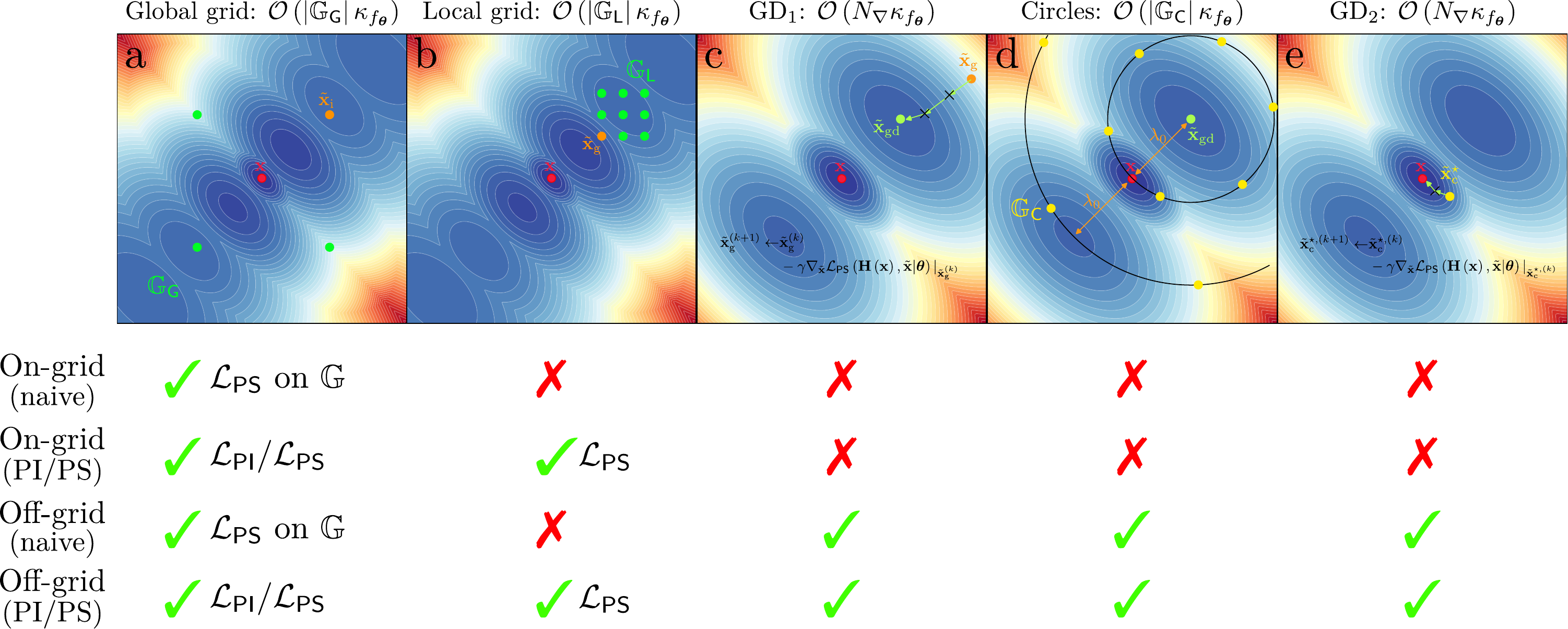}
        \caption{Illustration of the proposed localization scheme steps. $\mathrm{a}$: a grid-search is performed on the global grid $\Gt$. $\mathrm{b}$: a second grid-search is performed on the local grid $\Gl$, constructed from the obtained location at the previous step. $\mathrm{c}$: a first gradient-descent procedure is performed, and a \Redit{local minimum} is reached. $\mathrm{d}$: locations are sampled on circles of center $\tilde{\mathbf{x}}_{\mathrm{gd}}$ and radii $k \lambda_0$. $\mathrm{e}$: a second gradient-descent procedure is performed from the best location in the sampled circle locations, from a PS distance perspective. The bottom row illustrates the steps and loss functions involved in the proposed method (Off-Grid (PI/PS)) and its variants, which are presented in details in Section~\ref{sec:experiments}.}
        \label{fig:circle_explanation}
    \end{figure*}

    \begin{table}[t]
        \caption{Complexity comparison of the proposed approaches.}
        \centering
        \scriptsize
        \begin{tabular}{cccc}
            \toprule
            & On-grid (naive) & Off-grid (naive) & Off-grid (PI/PS)\\
            \midrule
            \makecell{Time \\complexity} & $\mathcal{O}\left(\abs{\Gn} \kappa_{\ftheta}\right)$ & $\makecell{\mathcal{O}\left(\left(\abs{\Gn} + \abs{\Gc} \right.\right. \\+ \left.\left.2N_{\nabla}\right) \kappa_{\ftheta}\right)}$ & $\makecell{\mathcal{O}\left(\left(\abs{\Gt} + \abs{\Gl} \right.\right.\\ + \left.\left.\abs{\Gc} + 2N_{\nabla}\right) \kappa_{\ftheta}\right)}$ \\
            \midrule
            \makecell{Forward \\passes nb.} & $218$k & $219.2$k & $11.2$k \\
            \bottomrule
        \end{tabular}
        \label{table:complexity_comp}
    \end{table}

    \noindent\textbf{Complexity optimization.} To mitigate the computational complexity of the proposed method, a bi-level grid approach is proposed. Initially, an exhaustive search is performed on a coarse grid, denoted as $\Gt \subset \mathbb{S}$, to determine an initial location estimate $\tilde{\mathbf{x}}_{\mathrm{i}}$. This \textit{global} grid is constructed based solely on the topological knowledge of the scene $\mathbb{S}$, which is used to define the grid boundaries. Subsequently, a finer \textit{local} grid, $\Gl \subset \mathbb{S}$, is generated around $\tilde{\mathbf{x}}_{\mathrm{i}}$, and a refined location estimate, $\tilde{\mathbf{x}}_{\mathrm{g}}$, is obtained through a grid search over $\Gl$. Finally, a gradient-descent optimization process is applied starting from this refined location estimate. The loss function used to obtain the initialization location, as defined in Eq.~\eqref{eq:init_grid_search}, can either be the Frobenius distance loss presented in Eq.~\eqref{eq:frob_sim_component} or \Redit{the phase-insensitive (PI)} distance introduced in~\cite{LeMagoarou21}, which is defined as follows:
    \begin{equation}\label{eq:PI_def}
        \mathcal{L}_{\mathsf{PI}}\left(\mathbf{H}\left(\mathbf{x}\right),\tilde{\mathbf{x}} \vert \boldsymbol{\theta}\right) \triangleq \sqrt{2-2\frac{\abs{\tr{\mathbf{H}\left(\mathbf{x}\right)^\herm \ftheta\left(\tilde{\mathbf{x}}\right)}}}{\norm{\mathbf{H}\left(\mathbf{x}\right)}{\mathsf{F}} \norm{\ftheta\left(\tilde{\mathbf{x}}\right)}{\mathsf{F}}}}.
    \end{equation}
    This PI distance does not distinguish maxima from minima due to the modulus operator in its definition. While this characteristic limits its applicability for achieving precise localization, it enhances the reliability of obtaining an estimate sufficiently close to the global minimum. Indeed, this inability to distinguish between minima and maxima results in a more gradual minimum profile, in contrast to the highly localized nature of the PS distance minima. The respective behaviors of the PS and PI distances for this localization problem will be illustrated in Section~\ref{sec:experiments}.

    A complexity comparison of the proposed approaches is presented in Table~\ref{table:complexity_comp}, where $\kappa_{\ftheta}$ represents the network forward pass complexity. The \textit{on-grid} approach refers to the localization without gradient-based refinement, whereas the \textit{off-grid} approach considers it. The PI/PS notation indicates which loss function is used for the grid-search on $\Gt$. As the bi-level grid approach allows to obtain a local grid center estimate that is close to the global minimum, it is possible to consider a low cardinality in $\Gl$, so that:
    \begin{equation}
        \abs{\Gt} + \abs{\Gl} + \abs{\Gc} + 2N_{\nabla} \ll \abs{\Gn}.
    \end{equation}
    Specifically, considering $\abs{\Gn} = 218.10^3$ (uniform grid on a $100$m$^2$ square scene, with $\lambda_0/4$ inter-element spacing), $\abs{\Gt} = 1.10^3$, $\abs{\Gl} = 9.10^3$, $\abs{\Gc} = 1.10^3$, and $N_{\nabla} = 100$ yields the second line in Table~\ref{table:complexity_comp}, which clearly highlights the overall complexity improvement obtained with the proposed approach.

\section{Experiments}\label{sec:experiments}

In this section, the localization performance of the proposed estimator is evaluated against traditional fingerprinting baselines such as $k$-NN, on realistic synthetic data in both indoor and outdoor scenes. 

\subsection{Evaluation framework}\label{sec:evaluation_framework}

\noindent\textbf{System parameters.} For all upcoming experiments, $N_s = 64$ frequencies are considered, with a central frequency $f_0 = 3.5$ GHz and a bandwidth of $20$ MHz. The associated central wavelength is $\lambda_0 \simeq 8.57$ cm. The BS is equipped with a uniform linear array composed of $N_a = 64$ antennas with half central wavelength spacing.

\noindent\textbf{Scene characteristics.} Three distinct scenes are considered, as represented in Fig.~\ref{fig:sionna_scene}: an outdoor urban environment in the \textit{Etoile} neighborhood of Paris, denoted as $\Sone$, and \Redit{two distinct indoor environments} with metallic obstacles, denoted as $\Stwo$ and $\Sthree$. The primary distinction between $\Stwo$ and $\Sthree$ is the higher proportion of locations without LoS paths in $\Sthree$ (around $50\%$), making it representative of a more complex radio environment. Additionally, in $\Sone$, the antenna array is located close to the UE locations, while it is far in $\Stwo$ and $\Sthree$. These contrasting settings enable the evaluation of the proposed localization method in scenarios representative of outdoor cellular deployments and indoor industrial environments. 

\begin{figure}[t]
    \centering
    \includegraphics[width=.95\columnwidth]{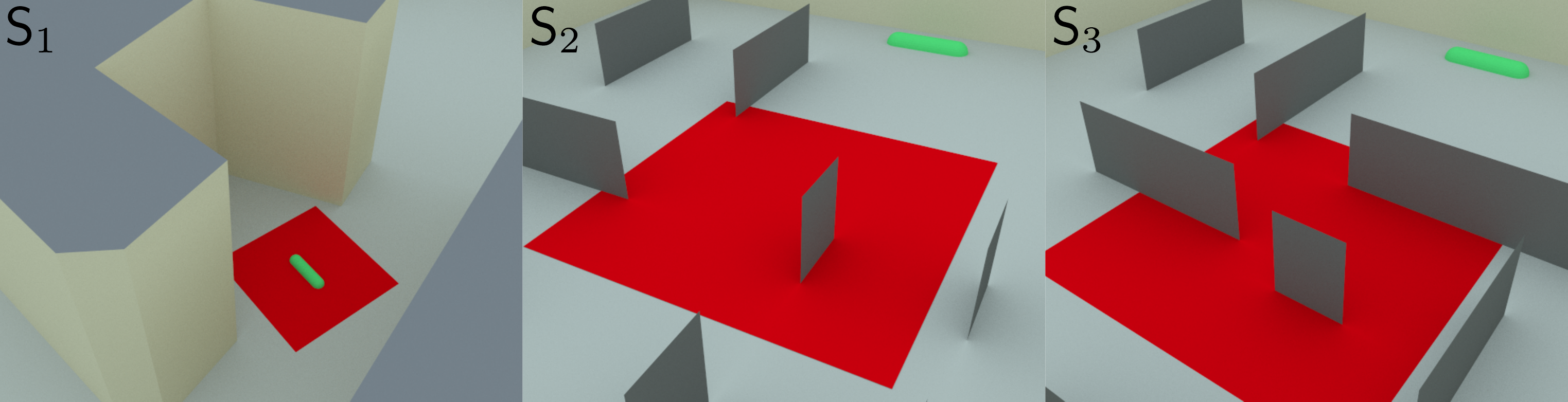}
    \caption{$3$D view of the considered scenes: the red plane represents the possible training/test/evaluation locations and the green spheres represent the BS antennas. In $\Stwo$ and $\Sthree$, gray planes represent metallic obstacles. For $\Sone$, only the locations in front of the BS are considered during the localization performance evaluation.}
    \label{fig:sionna_scene}
\end{figure}

\noindent\textbf{Location generation.} Each scene consists of a $L_{\mathbb{S},x} L_{\mathbb{S},y}$ m$^2$ plane considered as the location space $\mathbb{S}$. Within this location space, training, test, and evaluation locations are sampled. 

\begin{remark}
    \Redit{Note that obstacles in the propagation environment are taken into account during the sampling process by defining exclusion regions within which no locations can be sampled. This is formalized by introducing an exclusion location space $\mathbb{S}_{\texttt{e}}$. This approach enables the integration of accessibility constraints into the empirical validation framework, particularly for the $\Stwo$ and $\Sthree$ scenes, which include physical obstacles in the region of interest (denoted by the red plane in Fig.~\ref{fig:sionna_scene}).}
\end{remark}

\Redit{
The training locations, used to train the $\ftheta$ model, and denoted as $\mathcal{G}_{\mathsf{train}}$, are uniformly sampled as, $\forall \mathbf{x}_i \in \mathcal{G}_{\mathsf{train}}$:
\begin{equation}\label{eq:uniform_sampling}
    \hspace{-.399\baselineskip}\mathbf{x}_i \sim \mathcal{U}\left(\left(\left[-L_{\mathbb{S},x}/2, L_{\mathbb{S},x}/2\right] \times \left[-L_{\mathbb{S},y}/2, L_{\mathbb{S},y}/2\right]\right) \setminus \mathbb{S}_{\texttt{e}} \right).
\end{equation}
}
Test locations, denoted as $\mathcal{G}_{\mathsf{test}}$, are used to validate the mapping-learning capabilities of $\ftheta$. They form a uniform grid with an inter-element spacing of $\lambda_0/4$: \Redit{this yields a train/test ratio of around $7\%$.} Evaluation locations, collected in $\mathcal{G}_{\mathsf{loc}}$, are uniformly sampled within $\mathbb{S}$, following Eq.~\eqref{eq:uniform_sampling}, and are specifically used to assess the performance of the proposed localization method. Importantly, $\mathcal{G}_{\mathsf{train}} \cap \mathcal{G}_{\mathsf{test}} = \varnothing$, $\mathcal{G}_{\mathsf{train}} \cap \mathcal{G}_{\mathsf{loc}} = \varnothing$, and $\mathcal{G}_{\mathsf{test}} \cap \mathcal{G}_{\mathsf{loc}} = \varnothing$ ensuring that localization performance is evaluated on locations independent of those used for training and testing. Note that for $\Sone$, $L_{\mathbb{S},x} = 10$ m and $L_{\mathbb{S},y} = 5$ m, while $L_{\mathbb{S},x} = L_{\mathbb{S},y} = 10$ m for $\Stwo$ and $\Sthree$. Additionally, for all scenes, the training location density is set to $175$locs.$/$m$^2 \simeq 1.3$locs.$/\lambda_0^2$~\footnote{\Review{Results from our prior study~\cite{chatelier_loc2chan24} demonstrated effective mapping learning at a training location density of $0.5$ locs/$\lambda_0^2$, suggesting that satisfactory localization performance can be expected at lower densities.}}

\noindent\textbf{Channel generation.} The datasets are obtained through the radio-propagation digital twin of the Sionna library~\cite{sionna}. Since ray-tracing methods rely on fundamental electromagnetic principles, the resulting channels, both in indoor and outdoor scenes, can be considered realistic. 

\begin{remark}
    \Redit{Note that the ray-tracing approach is inherently limited by the accuracy of the radio-propagation digital twin. For instance, in a given scene, all buildings often share the same material, surfaces are typically modeled as unrealistically flat, and accurately capturing scattering effects induced by material roughness remains a challenging task.}
\end{remark}

For each antenna and location in the scene, ray-tracing techniques identify the propagation paths and compute the corresponding frequency-domain channel coefficients. The number of propagation paths depends on the location, with up to two consecutive reflections considered, while diffraction is neglected. The complete dataset consists of the locations and associated channel matrices, e.g. the training dataset is defined as:
\begin{equation}
    \mathcal{D}_{\mathsf{train}} = \left\{ \mathbf{x}_i \in \mathcal{G}_{\mathsf{train}}, \mathbf{H}\left(\mathbf{x}_i\right) \in \mathcal{H}_{\mathsf{train}} \right\}_{i=1}^{N_{\mathsf{train}}},
\end{equation}
where $\mathcal{H}_{\mathsf{train}}$ consists of the computed channel matrices at the training locations, and $N_{\mathsf{train}} = \abs{\mathcal{G}_{\mathsf{train}}}$ is the number of training locations.

\noindent\textbf{Evaluation metrics.} Let $\mathbf{x}$ denote the true UE location and $\hat{\mathbf{x}}$ its estimate. The performance of the localization methods is evaluated through statistical analysis of the localization error $\norm{\mathbf{x}-\hat{\mathbf{x}}}{2}$, using metrics such as the median, $10\%$ and $90\%$ quantiles.

\noindent\textbf{Baselines.} It is proposed to consider the weighted $k$-NN method on the training dataset as a fingerprinting baseline. Specifically, for a given channel $\mathbf{H}\left(\mathbf{x}\right)$, the $k$-best locations are computed as:
\begin{equation}
    \left\{\tilde{\mathbf{x}}_1, \cdots, \tilde{\mathbf{x}}_k\right\} = \underset{\tilde{\mathbf{x}} \in \mathcal{G}_{\mathsf{train}}}{\argmink} \frac{\norm{\mathbf{H}\left(\mathbf{x}\right)-\mathbf{H}\left(\tilde{\mathbf{x}}\right)}{\mathsf{F}}^2}{\norm{\mathbf{H}\left(\mathbf{x}\right)}{\mathsf{F}}^2},
\end{equation}
and the associated weights are, $\forall i \in \llbracket 1, k\rrbracket:$
\begin{equation}
    \alpha_i = \left(\frac{\norm{\mathbf{H}\left(\mathbf{x}\right)-\mathbf{H}\left(\tilde{\mathbf{x}}_i\right)}{\mathsf{F}}^2}{\norm{\mathbf{H}\left(\mathbf{x}\right)}{\mathsf{F}}^2}\right)^{-1},
\end{equation}
which are then normalized as $\tilde{\alpha}_i = \alpha_i/\sum_{i=1}^k \alpha_i$. The estimated location through $k$-NN is finally computed as:
\begin{equation}
    \hat{\mathbf{x}} = \sum_{i=1}^k \tilde{\alpha}_i \tilde{\mathbf{x}}_i.
\end{equation}

Considering this fingerprinting baseline on the training dataset, used to train the neural generative model, ensures a fair comparison with the proposed method, as both approaches use exactly the same data to generate a location estimate.

\subsection{Experimental results}\label{subsec:experimental_results}

Unless stated otherwise, in each scene, the following grid cardinalities are considered: $\abs{\Gt} = 1.10^3$, $\abs{\Gl} = 140.10^3$ for the PS init. method on the $\Stwo$ and $\Sthree$ scenes, and $\abs{\Gl} = 9.10^3$ for \Rours{both} the PS init. method on the $\Sone$ scene, and \Rours{the} PI init. method on all scenes. This cardinality difference in complex radio environments originates from the need for a larger coverage area in the PS case, which is attributed to the PS distance yielding a less accurate initial location estimate than the PI distance in such scenarios. Finally, $N_{\nabla} = 100$ gradient descent steps are considered.

\noindent\textbf{PI/PS distance behavior.} Fig.~\ref{fig:PI_PS_injectivity} illustrates the behavior of the PI and PS distances in the $\Sone$ scene for a given true location $\mathbf{x}$. It can be observed that the PS distance exhibits highly localized minima, which hinders the reliability of the initialization location estimate through the grid search on the global grid. In contrast, the PI distance exhibits a smoother, more gradual loss profile, with very high loss values anywhere except within a broad neighborhood in the true location direction, illustrated by the yellow area in the zoomed-out view of Fig.~\ref{fig:PI_PS_injectivity}. This facilitates the identification of an accurate initial location estimate through the global grid search. The closer spacing of minima in the PI distance, compared to the PS distance, can be attributed to its inability to distinguish maxima from minima, as previously discussed.



\noindent\textbf{Localization performance.} Fig.~\ref{fig:all_scene_perf} presents the performance of each approach by considering the boxplots of the localization error on $1.10^4$ independent locations. The low/high bars of the boxplots respectively represent the $10\%$ and $90\%$ quantiles, the central orange line represents the median, and the box boundaries correspond to the first and third quartiles. The different methods are as follows: 
\begin{itemize}
    \item MLP and $3$-NN: fingerprinting baselines. \Redit{The MLP has four complex hidden layers of sizes $4096$, $2048$, $1024$, and $512$, takes complex channel matrices $\mathbf{H}\left(\mathbf{x}\right)$ as inputs, and outputs a location estimate. It comprises approximately $22$ million learnable parameters, and is trained for $500$ epochs.}
    \item On-grid (PI/PS): proposed method using the PI or PS loss in the exhaustive search on $\Gt$, but without the gradient-descent refinement.
    \item Off-grid (PS w/o circle): proposed method using the PS loss in the exhaustive search on $\Gt$, with gradient-descent, but without the circles approach (see Fig.~\ref{fig:circle_explanation} $\mathrm{d}$).
    \item Off-grid (PI/PS): proposed method using the PI or PS loss in the exhaustive search on $\Gt$, with both gradient-descent and circles approaches.
\end{itemize}

The proposed  off-grid (PI) method consistently outperforms the fingerprinting baselines across all scenes, achieving an improvement of two to three orders of magnitude in the median localization performance. Notably, its median value reaches $0.01$ cm in $\Sone$, and approximately $0.06$ cm in $\Sthree$, demonstrating its strong performance throughout the entire scene, and across challenging radio environment\Rours{s}. It is worth noting that \Rours{these} values are well below the central wavelength, $\lambda_0 \simeq 8.57$ cm, highlighting the sub-wavelength accuracy of the proposed method. As anticipated, the off-grid approach significantly improves localization accuracy by eliminating the constraint of a fixed grid. Moreover, the median error values for the on-grid PS approach are approximately on the order of the central wavelength, highlighting the high sensitivity of the PS distance to local minima. In contrast, the PI distance demonstrates greater robustness to this issue, corroborating the insights from Eq.~\eqref{eq:PI_def} and Fig.~\ref{fig:PI_PS_injectivity}. For both initialization methods, the circle-based approach effectively mitigates this limitation, and enhances localization resolution. The degraded performance of the proposed method in the $\Stwo$ and $\Sthree$ scenes, compared to $\Sone$, can be attributed to the increased complexity of the radio environments simulated in $\Stwo$ and $\Sthree$. The presence of strong NLoS paths, and the complete absence of LoS paths for around $50\%$ of the locations in $\Sthree$ results in a much more complex loss landscape, making convergence to the global minimum more challenging. In this scenario, the circle-based approach, which is designed to mitigate the local minima issue, reaches its limit as the PS distance function exhibits increasingly chaotic minima. This can be seen with the superior localization performance of the off-grid (PI) approach on the $\Sthree$ scene when considering only locations with LoS paths (denoted as w. LoS only in Fig.~\ref{fig:all_scene_perf}), as compared to locations with only NLoS paths (denoted as w. NLoS only in Fig.~\ref{fig:all_scene_perf}). This chaotic minima behavior is illustrated in Fig.~\ref{fig:loss_landscape} with the loss landscape over the $\Sone$ and $\Sthree$ scenes.

\begin{figure}[t]
    \centering
    \includegraphics[width=.9\columnwidth]{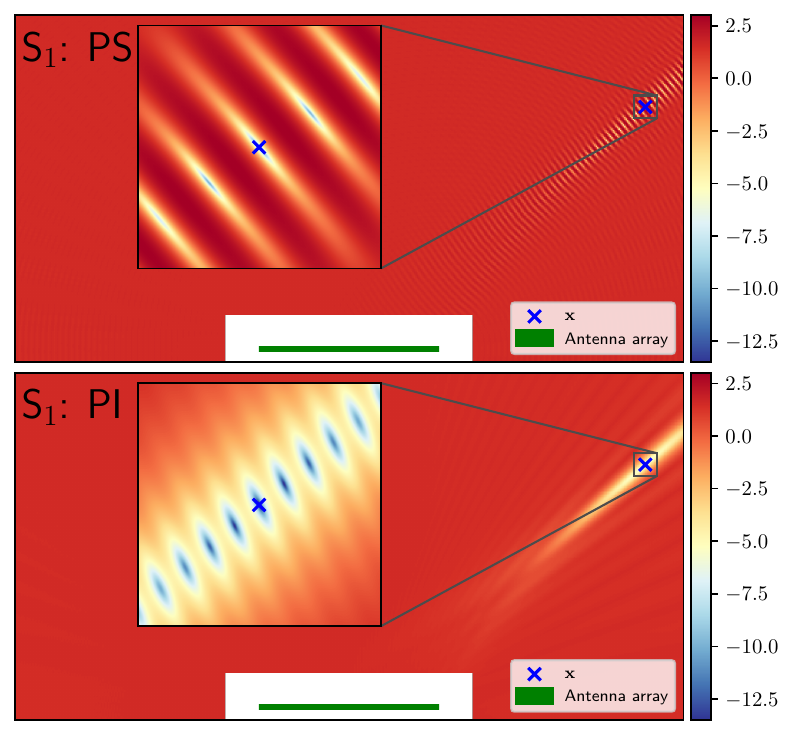}
    \caption{Comparison of $\mathcal{L}_{\mathsf{PS}}\left(\mathbf{H}\left(\mathbf{x}\right),\tilde{\mathbf{x}} \vert \boldsymbol{\theta}\right)$ and $\mathcal{L}_{\mathsf{PI}}\left(\mathbf{H}\left(\mathbf{x}\right),\tilde{\mathbf{x}} \vert \boldsymbol{\theta}\right)$, in dB, for a given true location $\mathbf{x}$ in the $\Sone$ scene.}
    \label{fig:PI_PS_injectivity}
\end{figure}

\begin{figure*}[htbp]
    \centering
    \begin{minipage}[b]{\textwidth}
        \centering
        \includegraphics[width=.9\textwidth]{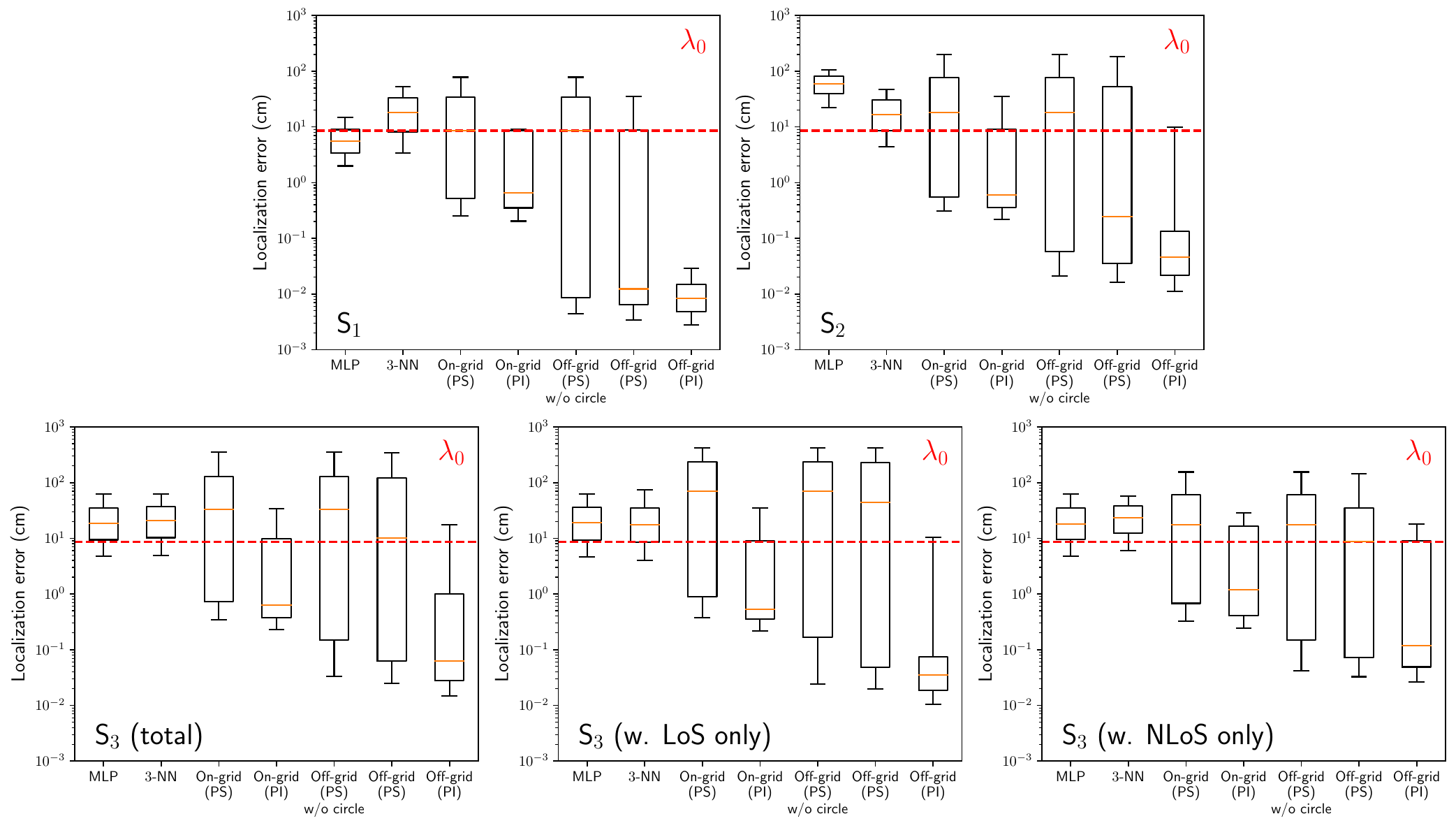}
        \caption{\Review{Statistics on the localization error (in cm) across the $\Sone$, $\Stwo$ and $\Sthree$ scenes. The red horizontal line represents the central wavelength $\lambda_0$.}}
        \label{fig:all_scene_perf}
    \end{minipage}
    \vfill
    \begin{minipage}[b]{0.49\textwidth}
        \centering
        \includegraphics[width=.99\columnwidth]{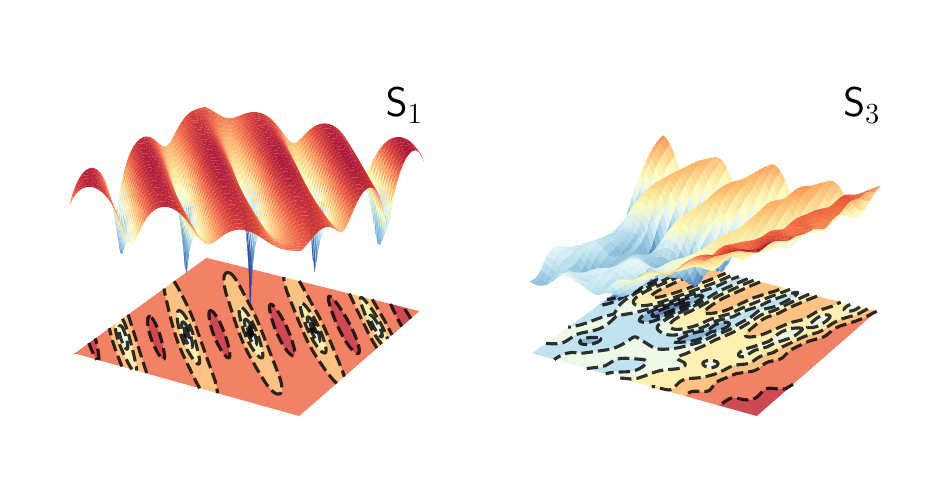}
        \caption{PS distance landscapes for the $\Sone$ and $\Sthree$ scenes.}
        \label{fig:loss_landscape}
    \end{minipage}
    \hfill
    \begin{minipage}[b]{0.49\textwidth}
        \centering
        \includegraphics[width=.9\columnwidth]{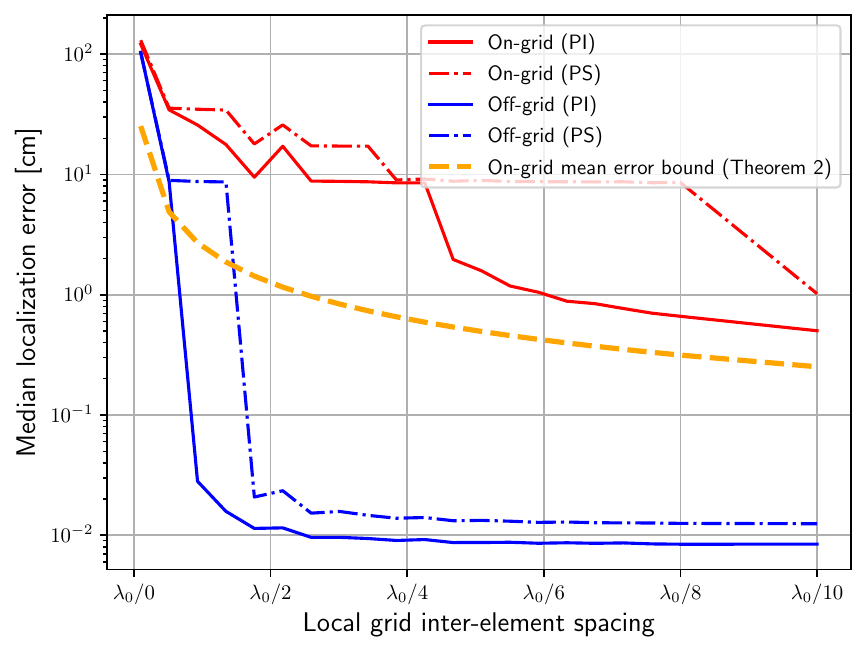}
        \caption{Localization performance evolution wrt. $\Gl$ inter-element spacing, over the $\Sone$ scene.}
        \label{fig:multi_lambda_div}
    \end{minipage}
    \vfill
    \begin{minipage}[b]{\textwidth}
        \centering
        \includegraphics[width=.95\textwidth]{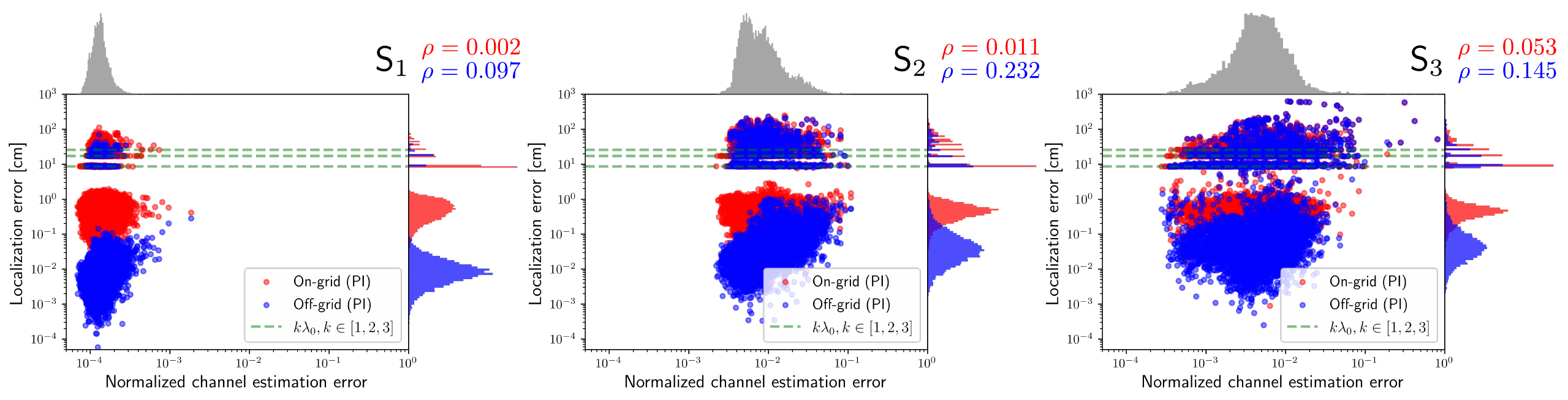}
        \caption{$2$D histogram of the location error wrt. the model error, across the $\Sone$, $\Stwo$ and $\Sthree$ scenes. A zero-valued Pearson's correlation coefficient $\rho$ indicates that the location and channel estimation errors are uncorrelated. Green dashed lines represent integer multiples of the central wavelength $\lambda_0$.}
        \label{fig:2D_histo}
    \end{minipage}

\end{figure*}

\noindent\textbf{Local grid spacing performance impact.} Fig.~\ref{fig:multi_lambda_div} illustrates the median localization error performance as a function of the inter-element spacing in $\Gl$, over the $\Sone$ scene. It is observed that the on-grid approach requires a very dense grid to approach the best-case bound defined in Theorem~\ref{thm:lambda_bound}. However, as anticipated, the off-grid approach surpasses this on-grid best-case performance bound, even with a relatively coarse grid. This demonstrates the efficiency of the off-grid approach in delivering accurate localization estimates, while maintaining a controlled computational complexity. This is enabled by the ability of the $\ftheta$ model to serve as a generative neural channel model with low estimation error.

\noindent\textbf{Relationship between model and localization performance.} \Review{In Fig.~\ref{fig:2D_histo}, each point in the bi-dimensional histograms corresponds, for a given true location, to the pair formed by the localization error and the channel estimation error, \Rours{i.e}. NMSE, at that true location, with marginals displayed on the right and top axes.} In the $\Sone$ scene, the near-zero Pearson's correlation coefficient indicates that the localization error is largely uncorrelated with the model error, whereas a noticeable correlation appears in the $\Stwo$ and $\Sthree$ scenes. This is likely due to the higher model NMSE in $\Stwo$ and $\Sthree$ (around $-10$ dB) compared to $\Sone$ (around $-30$dB). Nevertheless, even with high model NMSE, localization performance remains satisfactory. \Review{This implies that even when the reference channel, i.e. the channel estimated by the network at the true location, is imperfectly reconstructed, the resulting location estimate remains highly accurate, highlighting the robustness of the proposed method to noisy channels during inference.} Additionally, it can be observed that in $\Stwo$ and $\Sthree$, the off-grid approach benefits only at locations where the model error remains within acceptable limits. This behavior is consistent across all locations in $\Sone$, where the model error is uniformly low due to the less challenging nature of the scene. \Rours{Furthermore}, localization errors tend to accumulate at multiples of the central wavelength: this is an expected effect of the local minima in the PS distance, as discussed in Corollary~\ref{corol:minima_distance}. \Review{The off-grid approach, which incorporates the circles method, partially mitigates this local minima issue: owing to its ability to escape local minima through the circles sampling strategy, the proposed method can be viewed as a model-based annealing approach.}


\noindent\textbf{Computational complexity.} Table~\ref{table:time_comparison} reports the empirical inference computational complexity of the proposed methods and baselines over the $\Stwo$ scene. The average inference time for each method is computed over $1.10^4$ independent locations, with all \Rours{computations} performed on the same NVIDIA A$40$ GPU. The on-grid (naive) approach employs a uniform grid across the scene with an inter-element spacing of $\lambda_0/4$. The off-grid (naive) approach extends \Rours{it} by incorporating the proposed gradient descent and circle-based methods, but lacks the computational complexity optimization provided by the bi-level grid approach. Lastly, the off-grid (PI/PS) approaches follow Algorithm~\ref{alg:prop_approach}, with the $\Gl$ grid for the PS initialization being $16$ times larger than that for the PI initialization, which explains its superior mean execution time. While the \Rours{optimized method, i.e. off-grid (PI/PS),} exhibits higher computational complexity compared to the classical fingerprinting approach, it offers a substantial improvement in localization accuracy. Moreover, as demonstrated in Table~\ref{table:complexity_comp}, the proposed off-grid approach reduces the number of required forward passes by a factor of twenty with respect to the on-grid approach, resulting in an almost tenfold reduction in mean execution time. This is further illustrated by comparing the computational complexity of the proposed method with the off-grid naive approach, which experiences a high mean execution time due to the absence of computational complexity optimization.


\noindent\textbf{Memory requirements comparison.} As shown in Fig.~\ref{fig:all_scene_perf}, the $k$-NN fingerprinting approach, that uses the neural generative model training dataset, achieves relatively satisfactory performance, with sub-meter median localization accuracy across all scenes. However, this approach necessitates the storage of $2\abs{\mathcal{G}_{\mathsf{train}}} + 2N_a N_s\abs{\mathcal{G}_{\mathsf{train}}}$ real scalars, which, when evaluated on the $\Stwo$ training dataset, amounts to storing $121.3$M real scalars. In comparison, the proposed localization method only requires storing the $\abs{\boldsymbol{\theta}} = 9.1$M real learnable parameters of $\ftheta$, resulting in a thirteenfold reduction in the required memory. This number of learnable parameters originates from the hyperparameters selected for the MB-$\tilde{\mathbf{\Psi}}_{\mathbf{a}}$ network in~\cite{chatelier_loc2chan24}. Assuming a $32$-bit floating-point representation, the fingerprinting approach demands $485.2$ \Redit{Mb} of memory, whereas the proposed method only requires $36.4$ \Redit{Mb}. This substantial memory reduction is achieved by eliminating the need for a large fingerprinting comparison dictionary. Importantly, this memory efficiency does not compromise localization accuracy, as the proposed method is not constrained by the number of stored coefficients, unlike classical fingerprinting methods, as highlighted in Eq.~\eqref{eq:fingerprinting_resolution}. The independence of the proposed method’s localization performance from memory requirements is illustrated in Fig.~\ref{fig:memory_performance}. The proposed method consists of the off-grid (PI) approach, while the $3$-NN results for different memory requirements are obtained by subsampling $\mathcal{D}_{\mathsf{train}}$ on the $\Stwo$ scene.


\begin{table}[t]
    \caption{\Review{Average inference time over the $\Stwo$ scene.}}
    \centering
    \begin{tabular}{ccccccc}
        \toprule
        & \multicolumn{2}{c}{Baselines} & On-grid & \multicolumn{3}{c}{Off-grid}\\
        \cmidrule(lr){2-3} \cmidrule(lr){4-4} \cmidrule(lr){5-7} & \Review{MLP} & $3$-NN & Naive & Naive & PS & PI\\
        \midrule
        $\esp{t_{\mathsf{exec}}}$ [s] & \Review{$1.10^{-3}$} & $1.10^{-2}$ & $16.39$ & $17.72$ & $8.57$ & $1.75$ \\
        \bottomrule
    \end{tabular}
    \label{table:time_comparison}
\end{table}

\noindent\textbf{\Rours{Impact of noisy uplink channels during inference.}} \Rours{The proposed method assumes noiseless uplink channels during inference. To account for an imperfect sensing process, the performance of the proposed approach is also evaluated using noisy uplink channels at inference time. More specifically, the uplink channels used for performance evaluation are generated as follows:
\begin{equation}
    \tilde{\mathbf{H}}\left(\mathbf{x}\right) = \mathbf{H}\left(\mathbf{x}\right) + \mathbf{N},
\end{equation}
where:
\begin{equation}
    \mathbf{N} = \frac{\sigma}{\sqrt{2}}\left(\mathbf{Z}_1+\mathrm{j}\mathbf{Z}_2\right),
\end{equation}
with:
\begin{equation}
    \sigma = \sqrt{\frac{\norm{\mathbf{H}\left(\mathbf{x}\right)}{\mathsf{F}}^2}{N_a N_s \text{SNR}_{\text{lin}}}},
\end{equation}
and where each element of $\mathbf{Z}_1 \in \mathbb{R}^{N_a\times N_s}$ and $\mathbf{Z}_2 \in \mathbb{R}^{N_a\times N_s}$ are independently drawn for the $\mathcal{N}\left(0,1\right)$ distribution. Experimental results with noisy channels at $\text{SNR}_{\text{dB}} = 5$ dB are shown in Fig.~\ref{fig:noisy_channels}, where dark boxplots correspond to noise-free results, and blue ones to noisy results. It can be observed that the proposed approach demonstrates strong resilience to noise, as performance degradation remains minor across all methods and scenes. These results indicate that, even with an imperfect channel sensing process during the inference phase, the resulting impact on localization performance remains within acceptable limits.}

\begin{figure*}[!t]
    \centering
    \includegraphics[width=.95\textwidth]{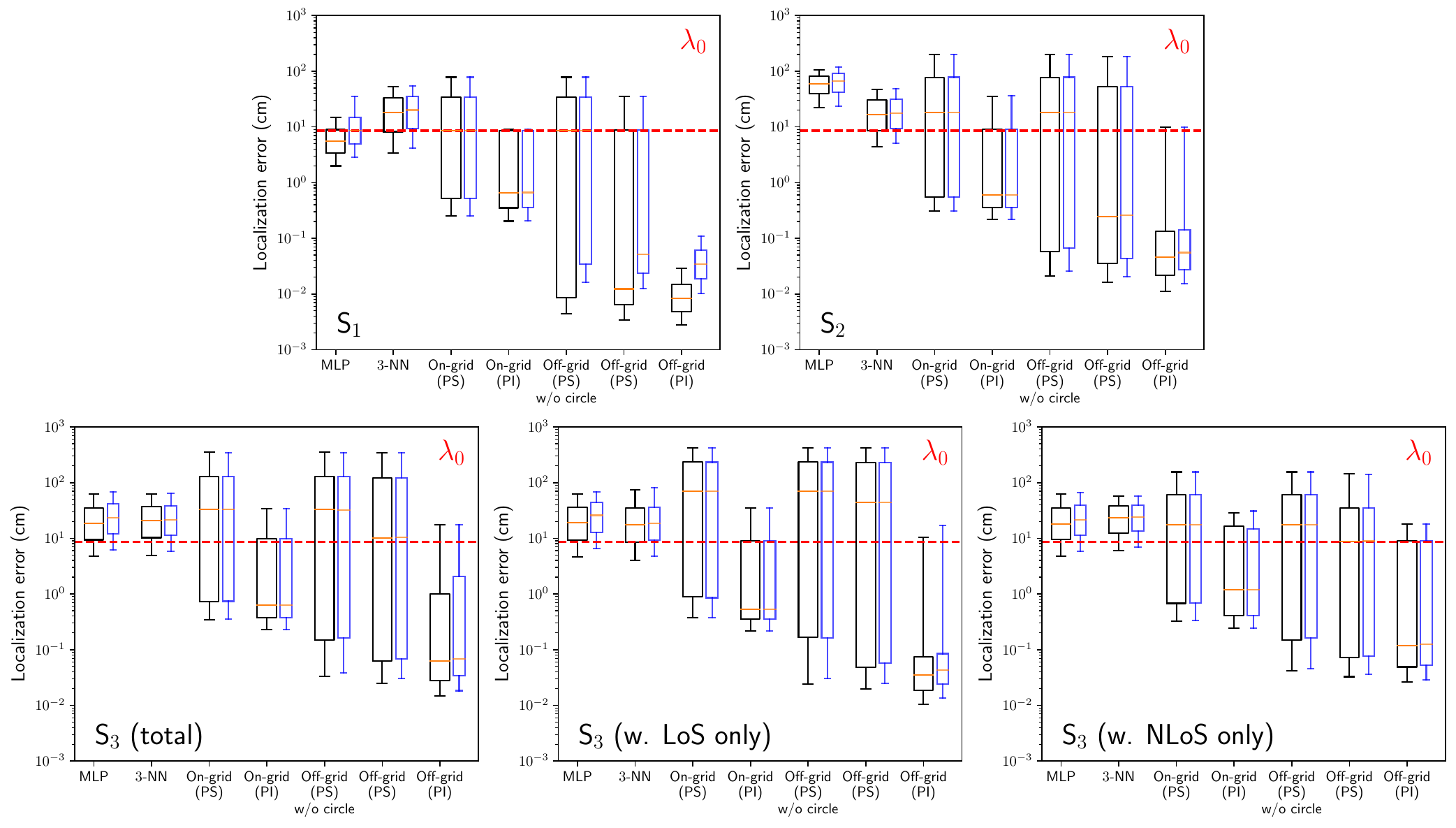}
    \caption{\Rours{Evaluation of noisy channels impact on the proposed localization method performance, during test time. Dark boxplots correspond to noise-free results, and blue ones to noisy results with $\text{SNR}_{\text{dB}} = 5$ dB.}}
    \label{fig:noisy_channels}
\end{figure*}

\noindent\textbf{\Redit{Impact of sparse training locations.}} \Redit{Hitherto, all empirical results assume a training location density of $175$locs./m$^2$ $\simeq 1.3$locs./$\lambda_0^2$. This implies that the location-to-channel neural network is trained on a highly spatially dense dataset. To investigate the impact of sparser training locations, Fig.~\ref{fig:loc_dens_impact} illustrates the performance of the proposed Off-grid (PI) localization method on the $\Stwo$ scene as the spatial density of the training dataset decreases. It can be seen that, as expected, the localization performance decreases with the training spatial density. Nevertheless, the proposed method maintains strong performance in the low-density regime. In particular, when considering a density of $0.18$locs./$\lambda_0^2$ $\simeq 25$locs./m$^2$, that is well below one training sample per squared wavelength, the localization performance remains satisfactory, achieving sub-wavelength median localization error.}

\begin{figure}[!t]
    \centering
    \includegraphics[width=.95\columnwidth]{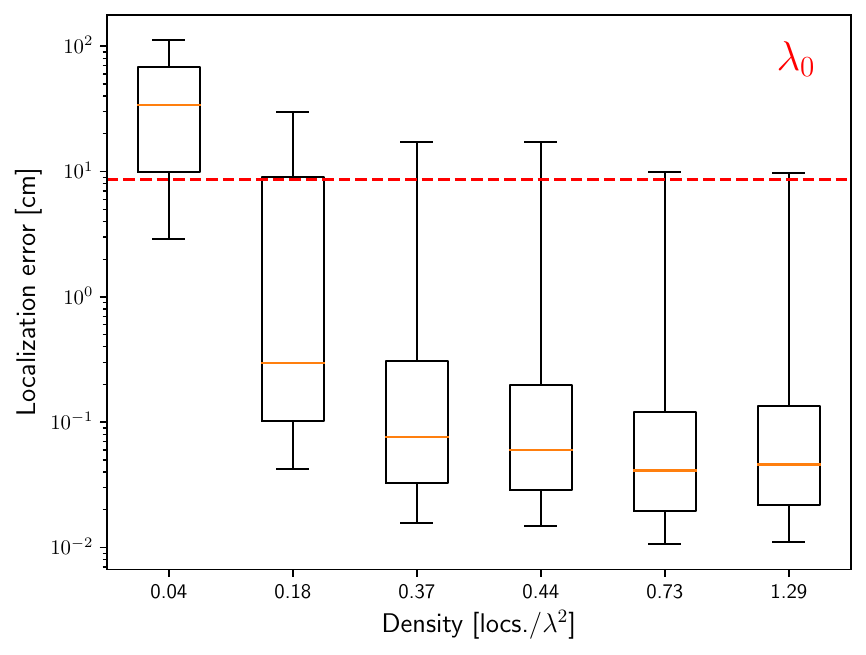}
    \caption{\Redit{Evaluation of sparse training location impact on the proposed method localization performance: Off-grid (PI) on $\Stwo$ scene.}}
    \label{fig:loc_dens_impact}
\end{figure}

\section{Conclusion and future work}\label{sec:conclusion}

This paper investigated the use of the model-based machine learning paradigm to perform user localization. Specifically, a neural architecture introduced in earlier studies, designed to learn the location-to-channel mapping, was leveraged to enhance the localization accuracy of the well-known fingerprinting method, while simultaneously reducing its memory requirements. Simulations conducted on realistic data, encompassing both LoS and NLoS challenging radio environments, demonstrated that the proposed method achieves sub-wavelength localization accuracy, providing up to a two-order-of-magnitude improvement over classical fingerprinting approaches. Unlike traditional fingerprinting methods, the performance of the proposed approach is not contingent on the size of a stored comparison dictionary. Instead, the neural model functions as a generative channel model, enabling on-the-fly inference of channel coefficients and thereby eliminating the need for an extensive precomputed dictionary. This resulted in a tenfold reduction in memory requirements compared to classical fingerprinting techniques. \Redit{In its current form, the proposed method could be used in dynamic environments to detect changes in the propagation environment. Indeed, since the method relies on a neural channel model trained for a fixed propagation environment, a significant modification of this environment, e.g. the introduction of a new obstacle generating strong additional propagation paths, would result in degraded channel estimation performance. Given a set of channel-location pairs generated by the neural channel model trained under nominal propagation conditions, environmental variations could be identified by monitoring the maximum similarity over time between a sensed channel and this reference set.}

\Redit{Future research directions can be organized into four main categories, encompassing both complexity/performance optimization, and adaptation of the proposed method to more realistic propagation environments.}

\noindent\textbf{\Redit{Computational complexity optimization.}} \Redit{To enable the use of the proposed method in practical systems, it is crucial to optimize its computational complexity. This can be achieved, for example, by refining the cardinalities of the various location grids.}

\noindent\textbf{\Redit{Performance optimization in NLoS environments.}} \Redit{Further optimization of the circle-based strategy in NLoS environments could be explored. In such scenarios, the presence of multipath propagation induces a more complex loss landscape and increases the risk of local minima. Enhancing robustness in these challenging conditions could further improve localization accuracy.}

\noindent\textbf{\Redit{Adaptation to time-varying propagation environments.}} \Redit{Extending the method to time-varying propagation environments could be achieved through online learning strategies. In such setups, continuously collected channel measurements could be used to iteratively refine the neural channel model, enabling its continuous adaptation over time. In this context, the domain-incremental learning approach for 5G indoor localization proposed in~\cite{Raichur26} appears particularly well suited, as it enables the rapid adaptation of neural models to environmental changes.}

\noindent\textbf{\Redit{Adaptation to realistic propagation environments.}} \Redit{An immediate research direction is the application of the proposed method to real measured channels, for instance using the Ultra Dense Indoor MaMIMO CSI dataset~\cite{Chenglong22}, or the CAEZ dataset~\cite{Wiesmayr25}. Additionally, the proposed method could be applied in propagation environments affected by hardware impairments, e.g. by incorporating a calibration procedure to mitigate their impact, as shown in~\cite{yassine2022,Chatelier2022,Mateos_Ramos24,chatelier24_diffMUSIC}. Finally, the method presented in this paper relies on ray-tracing to generate a channel coefficients dataset used to train a generative neural channel model. Since ray-tracing methods cannot fully capture all the underlying physical phenomena governing propagation channels, future work should aim at characterizing the modeling error between ray-tracing and measured channels in propagation environments of interest~\cite{Muthineni23,Bhatia23,Bhatia23b,Hoydis24learned}. The impact of this mismatch on the proposed method's performance should then be evaluated. In addition, hybrid approaches combining ray-tracing data and measured channels, e.g. transfer learning techniques, could be investigated to mitigate the potential impact of this mismatch.}

\bibliographystyle{IEEEtran}
\bibliography{refs/biblio.bib}

\appendices
\section{Proof of Theorem~\ref{thm:min_error}}\label{appendices:thm_min_error}
From Definition~\ref{def:ftheta_error}, it follows that:
\begin{align}
    \hspace{-.1\baselineskip}\forall \mathbf{x} \in \mathbb{S},\text{ } \Xi\left(\mathbf{x}\right) \mathop{\longrightarrow}_{\epsilon_{\ftheta} \rightarrow 0} 0 &\Rightarrow  \forall \mathbf{x} \in \mathbb{S},\text{ } \ftheta\left(\mathbf{x}\right) \mathop{\longrightarrow}_{\epsilon_{\ftheta} \rightarrow 0} \mathbf{H}\left(\mathbf{x}\right).
\end{align}
Assuming that Eq.~\eqref{eq:sim_measure} satisfies Definition~\ref{def:exact_injectivity}, and invoking Definition~\ref{def:estimator}, one obtains:
\begin{align}
    \forall \mathbf{x} \in \mathbb{S},\text{ } \hat{\mathbf{x}}\left(\mathbf{H}\left(\mathbf{x}\right) \vert \boldsymbol{\theta}, \mathbb{S}\right)&\mathop{\longrightarrow}_{\epsilon_{\ftheta} \rightarrow 0} \mathbf{x}, \notag \\
    \Rightarrow \sup_{\mathbf{x} \in \mathbb{S}}\norm{\mathbf{x}-\hat{\mathbf{x}}\left(\mathbf{H}\left(\mathbf{x}\right) \vert \boldsymbol{\theta}, \mathbb{S}\right)}{2}&\mathop{\longrightarrow}_{\epsilon_{\ftheta} \rightarrow 0} 0.
\end{align}

\section{Proof of Theorem~\ref{thm:lambda_bound}}\label{appendix:a}
Let $\epsilon \left(\mathbf{x}\right) \triangleq \norm{\mathbf{x}-\hat{\mathbf{x}}\left(\mathbf{H}\left(\mathbf{x}\right) \vert \boldsymbol{\theta}, \mathbb{G} \right)}{2}$. A lower bound on $\espo{\epsilon \left(\mathbf{x}\right)}{\mathbf{x}\sim \mathcal{P}_{\mathbf{x}}}$ can be derived by considering the best-case scenario, i.e. that the location estimate obtained by solving Eq.~\eqref{eq:mb_fingerprinting} on $\mathbb{G}$ is also optimal in the sense of the $\ell_2$ norm, which reads as:
\begin{align}
    \hat{\mathbf{x}}\left(\mathbf{H}\left(\mathbf{x}\right) \vert \boldsymbol{\theta}, \mathbb{G} \right) &= \argmin_{\tilde{\mathbf{x}} \in \mathbb{G}} \norm{\mathbf{H}\left(\mathbf{x}\right) - \ftheta\left(\tilde{\mathbf{x}}\right)}{\mathsf{F}}\notag\\
    &= \argmin_{\tilde{\mathbf{x}} \in \mathbb{G}} \norm{\mathbf{x}-\tilde{\mathbf{x}}}{2}.
\end{align}
By definition:
\begin{equation}
    \espo{\epsilon \left(\mathbf{x}\right)}{\mathbf{x}\sim \mathcal{P}_{\mathbf{x}}} = \int_{\mathbb{R}^2} \epsilon \left(\mathbf{x}\right) p_{\mathbf{x}} \left(x,y\right) \mathrm{d}x \mathrm{d}y, \label{eq:expectancy_error}
\end{equation}
where $p_{\mathbf{x}} \left(x,y\right)$ is the density function of $\mathcal{P}_{\mathbf{x}}$. Since under the considered assumptions, $\epsilon(\mathbf{x})$ denotes the distance between the true location and its closest point on the uniform grid $\Gn$ with inter-element spacing $\nu$, it can be assumed that $\mathbf{x}$ is uniformly distributed within a square of side length $\nu/2$, leading to:
\begin{equation}
    p_{\mathbf{x}} \left(x,y\right)= \frac{4}{\nu^2}\mathds{1}_{\left\{x\in\mathbb{R}^+ \vert x \leq \nu/2 \right\}}\left(x\right)\mathds{1}_{\left\{y\in\mathbb{R}^+ \vert y \leq \nu/2 \right\}}\left(y\right).
\end{equation}
Eq.~\eqref{eq:expectancy_error} can then be rewritten in polar coordinates as:
\begin{equation}
    \espo{\epsilon \left(\mathbf{x}\right)}{\mathbf{x}\sim \mathcal{P}_{\mathbf{x}}} \geq \frac{4}{\nu^2}\int_{0}^{\frac{\pi}{2}}\int_{0}^{r_m\left(\theta\right)} r^2 \mathrm{d}r \mathrm{d}\theta,
\end{equation}
where the axial integration bound is obtained from:
\begin{align}
    \begin{cases}
        x = \frac{\nu}{2}\\
        y = \frac{\nu}{2}
    \end{cases} &\Rightarrow
    \begin{cases}
        r_m\left(\theta\right)=\frac{\nu}{2\cos \theta}\\
        r_m\left(\theta\right)=\frac{\nu}{2 \sin \theta}
    \end{cases}\notag \\
    &\Rightarrow r_m\left(\theta\right) \triangleq \min\left(\frac{\nu}{2\cos \theta},\frac{\nu}{2 \sin \theta}\right).
\end{align}
One then obtains:
\begin{align}
    &\espo{\epsilon \left(\mathbf{x}\right)}{\mathbf{x}\sim \mathcal{P}_{\mathbf{x}}} \notag \\
    &\geq \frac{4}{\nu^2} \left(\int_{0}^{\frac{\pi}{4}}\int_{0}^{\frac{\nu}{2\cos\theta}} r^2 \mathrm{d}r \mathrm{d}\theta + \int_{\frac{\pi}{4}}^{\frac{\pi}{2}}\int_{0}^{\frac{\nu}{2\sin\theta}} r^2 \mathrm{d}r \mathrm{d}\theta \right)\notag\\
    &= \frac{\nu}{6}\left(\int_{0}^{\frac{\pi}{4}} \frac{1}{\cos^3 \theta} \mathrm{d}\theta + \int_{\frac{\pi}{4}}^{\frac{\pi}{2}} \frac{1}{\sin^3 \theta} \mathrm{d}\theta \right).
\end{align}
Recalling that:
\begin{equation}
    \begin{cases}
        \int \frac{1}{\cos^3 \theta} \mathrm{d}\theta = \frac{1}{2}\left(\sec \theta \tan \theta + \ln \left(\abs{\sec \theta + \tan \theta}\right)\right)\\
        \int \frac{1}{\sin^3 \theta} \mathrm{d}\theta = \frac{1}{2}\left(-\csc \theta \cot \theta + \ln \left(\abs{\csc \theta - \cot \theta}\right)  \right)
    \end{cases},
\end{equation}
yields the following lower bound:
\begin{align}
    \espo{\epsilon \left(\mathbf{x}\right)}{\mathbf{x}\sim \mathcal{P}_{\mathbf{x}}} &\geq \frac{\nu}{6}\left(\sqrt{2}+\ln\left(\sqrt{2}+1\right)\right),
\end{align}
so that, $\forall \mathbf{x} \in \mathbb{S}$:
\begin{equation}
    \espo{\norm{\mathbf{x}-\hat{\mathbf{x}}\left(\mathbf{H}\left(\mathbf{x}\right) \vert \boldsymbol{\theta}, \mathbb{G} \right)}{2}}{\mathbf{x}\sim \mathcal{P}_{\mathbf{x}}} \geq \nu \delta,
\end{equation}
with $\delta = \frac{1}{6}\left(\sqrt{2}+\ln\left(\sqrt{2}+1\right)\right)$, which concludes the proof. 

\section{Proof of Theorem~\ref{thm:minima_frob}}\label{appendix:c}
From Eq.~\eqref{eq:channel_model_virtual_sources}, introducing $d_{l,i}\left(\mathbf{x}\right) \triangleq \norm{\mathbf{x}-\mathbf{a}_{l,i}}{2}$ and $\phi_j \triangleq 2\pi/\lambda_j$ yields:
\begin{align}
    &\norm{\mathbf{H}\left(\mathbf{x}\right) - \mathbf{H}\left(\mathbf{x} + \bdelta\right)}{\mathsf{F}} \notag\\
    &= \sqrt{\sum_{i=1}^{N_a} \sum_{j=1}^{N_s}\abs{\sum_{l=1}^{L_p} \gamma_l \left(\dfrac{\ej[-]{\phi_j d_{l,i}\left(\mathbf{x}\right)}}{d_{l,i}\left(\mathbf{x}\right)} - \dfrac{\ej[-]{\phi_j d_{l,i}\left(\mathbf{x}+\bdelta\right)}}{d_{l,i}\left(\mathbf{x}+\bdelta\right)} \right)}^2}\notag\\
    &= \sqrt{\sum_{i=1}^{N_a} \sum_{j=1}^{N_s}\abs{\sum_{l=1}^{L_p} f_{i,j,l}\left(\bdelta\right)}^2}.\label{eqtemp:chan_diff}
\end{align}
Considering $\norm{\bdelta}{2} \leq \epsilon_{\bdelta}$, so that $d_{l,i} \left(\mathbf{x}\right)/ d_{l,i}\left(\mathbf{x}+\bdelta\right) \simeq 1$, $f_{i,j,l}\left(\bdelta\right)$ can be rewritten as:
\begin{align}
    f_{i,j,l}\left(\bdelta\right) &\simeq  \frac{\gamma_l}{d_{l,i}\left(\mathbf{x}\right)} \left(\ej[-]{\phi_j d_{l,i}\left(\mathbf{x}\right)} - \ej[-]{\phi_j d_{l,i}\left(\mathbf{x} + \bdelta\right)}\right).
\end{align}

The function $\norm{\mathbf{H}\left(\mathbf{x}\right)- \mathbf{H}\left(\mathbf{x}+\bdelta\right)}{\mathsf{F}}$ being lower-bounded by $0$, a sufficient condition to obtain its global minimum is given by:
\begin{align}
    \bdelta &\in \bigcap_{\left(i,j,l\right) \in \mathbb{I}} \left\{\bdelta \in \mathbb{R}^2 \vert f_{i,j,l}\left(\bdelta\right) = 0 \right\}\notag\\
    \Leftrightarrow \bdelta &\in \bigcap_{\left(i,j,l\right) \in \mathbb{I}} \left\{\bdelta \in \mathbb{R}^2 \vert \ej[-]{\phi_j d_{l,i}\left(\mathbf{x}\right)} = \ej[-]{\phi_j d_{l,i}\left(\mathbf{x}+\bdelta\right)} \right\},
\end{align}
with $\mathbb{I} = \llbracket 1,N_a\rrbracket \times \llbracket 1,N_s\rrbracket \times \llbracket 1,L_p\rrbracket$. Let $k\in\mathbb{Z}$, then:
\begin{align}
    &\ej[-]{\phi_j d_{l,i}\left(\mathbf{x}\right)} = \ej[-]{\phi_j d_{l,i}\left(\mathbf{x}+\bdelta\right)} \notag\\
    &\Leftrightarrow \phi_j d_{l,i}\left(\mathbf{x}\right) = \phi_j d_{l,i}\left(\mathbf{x}+\bdelta\right) + k 2\pi \notag\\
    &\Leftrightarrow \norm{\mathbf{x} - \mathbf{a}_{l,i} + \bdelta}{2} = \norm{\mathbf{x} - \mathbf{a}_{l,i}}{2} + k\lambda_j \notag\\
    &\Leftrightarrow \mathbf{x} + \bdelta \in \mathcal{C}\left(\mathbf{a}_{l,i},\norm{\mathbf{x}-\mathbf{a}_{l,i}}{2} + k \lambda_j\right).
\end{align}
The global-minimum condition is thus given by:
\begin{equation}\label{eqtemp:condition}
    \bdelta \in \bigcap_{\left(i,j,l,k\right) \in \mathbb{L}} \mathcal{C}\left(\mathbf{a}_{l,i}-\mathbf{x},\norm{\mathbf{x}-\mathbf{a}_{l,i}}{2} + k \lambda_j\right),
\end{equation}
with $\mathbb{L} = \llbracket 1,N_a\rrbracket \times \llbracket 1,N_s\rrbracket \times \llbracket 1,L_p\rrbracket \times \mathbb{Z}$. Note that, in the absence of ambiguity, the circles intersect exclusively at $0_{\mathbb{R}^2}$ for $k = 0$, implying that $\boldsymbol{\delta} = 0_{\mathbb{R}^2}$. When $\boldsymbol{\delta} = 0_{\mathbb{R}^2}$, $\norm{\mathbf{H}(\mathbf{x}) - \mathbf{H}(\mathbf{x} + \boldsymbol{\delta})}{\mathsf{F}}$ attains its minimum possible value, which confirms that Eq.~\eqref{eqtemp:condition} characterizes a global minimum.

\section{Proof of Corollary~\ref{corol:approx_injectivity}}\label{appendix:d}
From Eq.~\eqref{eq:sim_measure}, the similarity measure is maximized when the Frobenius distance is minimized. Hence, satisfying Definition~\ref{def:exact_injectivity} reduces to showing that $\boldsymbol{\delta} = 0_{\mathbb{R}^2}$ is the global minimum of $\norm{\mathbf{H}(\mathbf{x}) - \mathbf{H}(\mathbf{x} + \boldsymbol{\delta})}{\mathsf{F}}$, which was established in Theorem~\ref{thm:minima_frob} under the no-ambiguity assumption.

\section{Proof of Corollary~\ref{corol:minima_distance}}\label{appendix:e}
Assuming a dominant propagation path is equivalent to considering that:
\begin{equation}
    \exists l_d \in \llbracket 1,L_p \rrbracket, \forall l \in \llbracket 1,L_p \rrbracket, \abs{\gamma_{l_d}} \gg \abs{\gamma_l}.
\end{equation}
This translates into negligible path components in Eq.~\eqref{eqtemp:chan_diff} so that:
\begin{align}
    &\norm{\mathbf{H}\left(\mathbf{x}\right)- \mathbf{H}\left(\mathbf{x}+\bdelta\right)}{\mathsf{F}} \notag \\
    &\hspace{-.5\baselineskip}\simeq \sqrt{\sum_{i=1}^{N_a}\sum_{j=1}^{N_s}\abs{\gamma_{l_d} \left(\dfrac{\ej[-]{\phi_j d_{l_d,i}\left(\mathbf{x}\right)}}{d_{l_d,i}\left(\mathbf{x}\right)} - \dfrac{\ej[-]{\phi_j d_{l_d,i}\left(\mathbf{x}+\bdelta\right)}}{d_{l_d,i}\left(\mathbf{x}+\bdelta\right)} \right)}^2}.
\end{align}
The minima condition defined in Eq.~\eqref{eq:thm_condition_global_min} can then be rewritten as:
\begin{equation}\label{eqtemp:intersec}
    \bdelta \in \bigcap_{\left(i,j,k\right) \in \mathbb{J}} \mathcal{C}\left(\mathbf{a}_{l_d,i}-\mathbf{x},\norm{\mathbf{x} - \mathbf{a}_{l_d,i}}{2}+k \lambda_j\right),
\end{equation}
with $\mathbb{J} = \llbracket 1,N_a\rrbracket \times \llbracket 1,N_s\rrbracket \times \mathbb{Z}$. The minima thus lie on circles centered at the locations of the antenna array elements, with radii determined by the system's wavelengths. In conventional communication systems, wavelength variation across the bandwidth is negligible, such that $\forall j \in \llbracket 1, N_s \rrbracket$:
\begin{equation}
    \lambda_j \simeq \lambda_0 \triangleq \frac{1}{N_s} \sum_{i=1}^{N_s} \frac{c}{f_i},
\end{equation}
which yields:
\begin{equation}\label{eqtemp:intersec_approx}
    \bdelta \in \bigcap_{\left(i,k\right) \in \mathbb{K}} \mathcal{C}\left(\mathbf{a}_{l_d,i}-\mathbf{x},\norm{\mathbf{x} - \mathbf{a}_{l_d,i}}{2}+k \lambda_0\right),
\end{equation}
with $\mathbb{K} = \llbracket 1,N_a\rrbracket \times \mathbb{Z}$. In Eq.~\eqref{eqtemp:intersec_approx}, the intersection of the circles originating from the different antennas depends only on the chosen integer $k \in \mathbb{Z}$, which translates into consecutive minima being approximately spaced by $\lambda_0$.

\end{document}